\documentclass[acmsmall,screen]{acmart}
\pdfoutput=1

\setcopyright{rightsretained}
\acmPrice{}
\acmDOI{10.1145/3360554}
\acmYear{2019}
\copyrightyear{2019}
\acmJournal{PACMPL}
\acmVolume{3}
\acmNumber{OOPSLA}
\acmArticle{128}
\acmMonth{10}

\usepackage{booktabs}   
\usepackage{subcaption} 
\usepackage{pgfplots}
\usetikzlibrary{patterns}
\pgfplotsset{width=5.5cm,compat=1.9}
\newtheorem{theorem}{Theorem}[section]
\newtheorem{claim}[theorem]{Claim}

\usepackage{listings}
\usepackage{hyperref}
\usepackage{xcolor}
\newcommand{\code}[1]{{\texttt{#1}}}
\newcommand{\barwidth}{3.5pt}
\newcommand{\ynote}[1]{}
\newcommand{\ignore}[1]{}

\begin{document}
\newcommand{\TITLE}{\text{Efficient Lock-Free Durable Sets}}
\newcommand{\SOFT}{\textsc{soft}}
\newcommand{\maxImprovement}{3.3x}
\definecolor{LogFree}{RGB}{1, 66, 37}
\definecolor{LinkFree}{RGB}{191, 0, 255}
\definecolor{SOFT}{RGB}{255, 88, 0}
\title[\TITLE]{\TITLE}         
\titlenote{This work was supported by the Israel Science Foundation grant No. 274/14}             


\author{Yoav Zuriel}
\affiliation{
  \department{CS Department}              
  \institution{Technion}            
  \country{Israel}                    
}
\email{yoavzuriel@cs.technion.ac.il}          

\author{Michal Friedman}
\affiliation{
  \department{CS Department}             
  \institution{Technion}           
  \country{Israel}                   
}
\email{michal.f@cs.technion.ac.il}         

\author{Gali Sheffi}
\affiliation{
  \department{CS Department}             
  \institution{Technion}           
  \country{Israel}                   
}
\email{galish@cs.technion.ac.il}         

\author{Nachshon Cohen}
\affiliation{
  \institution{Amazon}           
  \country{Israel}                   
}
\email{nachshonc@gmail.com}         
\authornote{Work done in cooperation with the Technion (external and not related to the author's work at Amazon, but done during employment time).}          

\author{Erez Petrank}
\affiliation{
  \department{CS Department}             
  \institution{Technion}           
  \country{Israel}                   
}
\email{erez@cs.technion.ac.il}         

\begin{abstract}
Non-volatile memory is expected to co-exist or replace DRAM in upcoming architectures.
Durable concurrent data structures for non-volatile memories are essential building blocks for constructing adequate software for use with these architectures.
In this paper, we propose a new approach for durable concurrent sets and use this approach to build the most efficient durable hash tables available today.
Evaluation shows a performance improvement factor of up to \maxImprovement{} over existing technology.
\end{abstract}



\maketitle

\section{Introduction}\label{sec.intro}
An up-and-coming innovative technological advancement is non-volatile RAM (NVRAM).
This new memory architecture combines the advantages of DRAM and SSD.
The latencies of NVRAM are expected to come close to DRAM, and it can be accessed at the byte level using standard \code{store} and \code{load} operations, in contrast to SSD, which is much slower and can be accessed only at a block level.
Unlike DRAM, the storage of NVRAM is persistent, meaning that after a power failure and a reset, all data written to the NVRAM is saved~\citep{Zhang:7208275}.
That data, in turn, can be used to reconstruct a state similar to the one before the crash, allowing continued computation.

Nevertheless, it is expected that caches and registers will remain volatile~\citep{10.1007/978-3-662-53426-7_23}.
Therefore, the state of data structures underlying standard algorithms might not be complete in the NVRAM view, and after a crash this view might not be consistent because of missed writes that were in the caches but did not reach the memory.
Moreover, for better performance, the processor may change the order in which writes reach the NVRAM, making it difficult for the NVRAM to even reflect a consistent prefix of the computation. 
In simpler words, the order in which values are written to the memory may be different from the program order.
Thus, the implementations and the correctness conditions for programs become more involved. 

Harnessing durable storage requires the development of new algorithms that can ensure a consistent state of the program in memory when a crash occurs and the development of corresponding recovery mechanisms.
These algorithms need to write back cache lines explicitly to the NVRAM, to ensure that important stores persist in an adequate order. 
The latter can be obtained using a \code{FLUSH} instruction that explicitly writes back cache lines to the DRAM.
Flushes typically need to be accompanied by a memory fence in order to guarantee that the write back is executed before continuing the execution. This combination of instructions is denoted \code{psync}.
The cost of flushes and memory fences is high, hence their use should be minimized to improve performance.

When dealing with concurrent data structures, \emph{linearizability} is often used as the correctness definition \citep{Herlihy:1990:LCC:78969.78972}.
An execution is \emph{linearizable} if every operation seems to take effect instantaneously at a point between its invocation and response. 
Various definitions of correctness for durable algorithms have been proposed.
These definitions extend linearizability to the setting that includes crashes, recoveries, and flush events.
In this work, we adopt the definition of \citet{10.1007/978-3-662-53426-7_23} denoted {\em durable linearizability}. 
Executions in this case also include crashes alongside invocations and responses of operations.
Intuitively, an execution is \emph{durable linearizable} if all operations that survive the crashes are linearizable.

This work is about implementing efficient set data structures for NVRAM.
Sets (most notably hash maps) are widely used, e.g., for key-value storage \citep{nishtala2013scaling, raju2017pebblesdb, Debnath:2010:FHT:1920841.1921015}.
It is, therefore, expected that durable sets would be of high importance when NVRAMs reach mass production. 
The durable sets proposed in this paper are the most efficient available today and can yield better throughput for systems that require fault-tolerance.
Our proposed data structures are all lock-free, which make them particularly adequate for the setting. First, lock-free data structures are naturally efficient and scalable \citep{herlihy2011art}.
Second, the use of locks in the face of crashes requires costly logging to undo instructions executed in a critical section that did not complete before the crash.
Nesting of locks may complicate this task substantially \citep{Chakrabarti:2014:ALL:2714064.2660224}.

State-of-the-art constructions of durable lock-free sets, denoted \emph{Log-Free Data Structures}, were recently presented by~\citet{David:215967}.
They proposed two clever techniques to optimize durable structures and built four implementations of sets.
Their techniques were aimed at reducing the number of required explicit write backs (\code{psync} operations) to the non-volatile memory.

In this paper, we present a new idea with two algorithms for durable lock-free sets, 
which reduce the required flushes substantially.
Whereas previous work attempted to reduce flushes that were not absolutely necessary for recovery, we propose to completely avoid persisting any pointer in the data structure.
In a crash-free execution, we can use the pointers to access data quickly, but when a crash occurs, we do not need to access a specific key fast.
We only need a way to find all nodes to be able to decide which belong to the set and which do not. 
This idea is applicable to a set because for a set we only care if a node (which represents a key) belongs to the data structure or not.
Thus, we only persist the nodes that represent set members by flushing their content to the NVRAM, but we do not worry about persisting pointers that link these nodes -{}- hence the name {\em link-free}. 
The persistent information on the nodes allows determining (after a crash) whether a node belongs to the set or not.
We also allow access to all potential data structure nodes after a crash so that during recovery we can find all the members of the set and reconstruct the set data structure.
We do that by keeping all potential set nodes in special designated areas, which are accessible after a crash. 

In volatile memory, we still use the original pointers of the data structure to allow fast access to the set nodes, e.g., by keeping a hash map (in the volatile memory) that allows fast access to members of the set.
Not persisting pointers significantly reduces the number of flushes (and associated fences), 
thereby, drastically improving the performance of the obtained durable data structure.
To recover from a crash, the recovery algorithm traverses all potential set nodes to determine which belong to the set.
The recovery procedure reconstructs the full set data structure in the volatile space, enabling further efficient computation.

The first algorithm that we propose, called {\em link-free}, implements the idea outlined in the above discussion in a straightforward manner.
The second algorithm, called \SOFT{}, attempts to further reduce the number of fences to the minimum theoretical bound.
This achievement comes at the expense of algorithmic complication.
Without flushes, the first (link-free) algorithm would probably be more performant, as it executes fewer instructions.
Nevertheless, in the presence of flushes and fences, the second (\SOFT{}) algorithm often outperforms link-free.
Interestingly, \SOFT{} executes at most one fence per thread per update operation.
It has been shown in \citep{Cohen:3210377.3210400} that there are no durable data structures that can execute fewer fences in the worst case.
Thus, \SOFT{} matches the theoretical lower bound, and is also efficient in practice.

On top of the innovative proposal to avoid persisting pointers (and its involved implementation), 
we also adopt many clever techniques from previous work. 
Among them, we employ the link-and-persist technique from~\citet{David:215967} that uses a flag to signify that an address has already been flushed so that further redundant \code{psync} operations can be avoided.  
Another innovative technique follows an observation in~\citet{Cohen:3152284.3133891} that flushes can be elided when writing several times to the same cache line. 
In such case, it is sufficient to use fences (or, on a TSO platform, only compiler fences) to ensure the order of writes to cache and the same order is guaranteed also when writing to the NVRAM.
Each write back of this cache line to the memory always reflects a prefix of the writes as executed on the cache line. 

Both schemes are applicable to linked lists, hash tables, skip lists and binary search trees and both guarantee lock-freedom and maintain a consistent state upon a failure.
We implemented a basic durable lock-free linked list and a durable lock-free hash table based on these two schemes and evaluated them against the durable lock-free linked list and hash map of \citet{David:215967}.
The code for these implementations is publicly available in GitHub at \href{https://github.com/yoavz1997/Efficient-Lock-Free-Durable-Sets}{https://github.com/yoavz1997/Efficient-Lock-Free-Durable-Sets}.
Our algorithms outperform previous state-of-the-art durable hash maps by a factor of up to \maxImprovement{}.

The basic assumption in this work (as well as previous work mentioned) is that crashes are infrequent, as is the case for servers, desktops, laptops, smartphones, etc.
Therefore, efficiency is due to low overhead on data structures operation.
The algorithms proposed here do not fit a scenario where crashes are frequent.
Substantial work on dealing with scenarios in which crashes are frequent has been done.
The research focuses on energy harvesting devices in which power failures are an integral part of the execution, e.g.,~\citep{van2016intermittent, maeng2017alpaca, colin2016chain, lucia2017intermittent, ruppel2019transactional,  maeng2018adaptive, jayakumar2015quickrecall, yildirim2018ink}. Some of these devices also have a non-volatile memory (FRAM) and volatile registers. To deal with the
frequent crashes, programs are executed by using checkpoints (enforced by the programmer, by the compiler, by run time, or by special hardware), and thus achieve persistent execution.
Currently, those approaches do not deal with concurrency or with durable linearizability.

The rest of the paper is organized as follows.
In Section~\ref{sec.set_overview} we provide an overview of the set algorithms.
Sections~\ref{sec.lf_details} and \ref{sec.soft_details} provide the details of the \emph{link-free} and \emph{\SOFT{}} algorithms, respectively.
In Section~\ref{sec.mm} we discuss the memory management scheme used.
The evaluation is laid out in Section~\ref{sec.results}.
Finally, we examine related work in Section~\ref{sec.related_work}, and conclude in Section~\ref{sec.conclusion}. 
Formal correctness proofs for the link-free list and the \SOFT{} list are laid out in Appendices 
~\ref{chap:link_free_correctness_gali} and~\ref{chap:soft_correctness_gali}.

\section{Overview of The Proposed Data Structures}\label{sec.set_overview}
A \emph{set} is an abstract data structure that maintains a collection of unique keys.
It supports three basic operations: \emph{insert}, \emph{remove}, and \emph{contains}.
The \emph{insert} operation adds a key to the set if the key is not already in the set. The \emph{remove} operation deletes the given key from the set (if the key belongs to the set) and the \emph{contains} operation checks whether a given key is in the set.
A key in a set is usually associated with some data.
In our implementation we assume this data comprise one word.
Our scheme can easily be extended to support other forms of data or no data at all.

A typical implementation of a lock-free set relies on a lock-free linked graph, such as a linked list, a skip list, a hash table, or a binary search tree (e.g., \citep{Harris:10.1007/3-540-45414-4_21, shalev2006split, natarajan2014fast, herlihy2011art, Michael:564870.564881}). 
Each node typically represents a single key and consists of a key, a value, and a \emph{next} pointer(s) to one (or more) additional nodes in the set. 
The structure of the linking pointers determines the set complexity, from a simple linked list (i.e., a single \emph{next} pointer) to skip lists or binary search trees. 

One way to transform a lock-free set into a durable one\footnote{When saying an algorithm is durable we mean the algorithm is durable linearizable~\citep{10.1007/978-3-662-53426-7_23}.} is to ensure that the entire structure is kept consistent in the NVRAM \citep{10.1007/978-3-662-53426-7_23}.
Using this method, each modification to the set has to be written immediately to the NVRAM. 
When reading from the set, readers are also required to flush the read content, to avoid acting according to values that would not survive a crash. 
Upon recovery, the content of the data structure in the non-volatile memory matches a consistent prefix of the execution.
The problem with this approach is that the large number of flushes imposes a high performance overhead. 

In this paper, we take a different approach that fits data structures that represent sets. 
Instead of keeping the entire structure in NVRAM, we only ensure that the key and the value of each node are stored durably. 
In addition, we maintain a persistent state in each node, which lets the recovery procedure determine whether the insertion of a specific node has been completed and whether this node has not been removed.
By providing such per-node information, we avoid needing to keep the linking structure (i.e., \emph{next} pointers) of the set.

Both of our set algorithms maintain a basic unit called the \emph{persistent node}, consisting of a key, a value and a Boolean method for determining whether the key in the node is a valid member of the set.
The persistent nodes are allocated in special \emph{durable areas}, which only contain persistent nodes. 
During execution, the system manages a collection of durable areas from which persistent nodes are allocated. 
Following a crash, the recovery procedure iterates over the durable areas and reconstructs the data structure with all its volatile links from all valid nodes. 

A major challenge we face in the design of our algorithms is to ensure that the order in which operations take effect in the non-volatile memory view matches some linearization order of the operations executed in the volatile memory.
This match is required to guarantee the durable linearizability of the algorithms.

One standard techniques employed in the proposed algorithms is the marking of nodes as {\em removed} by setting the least significant bit of one of the node's pointers.
This method was presented by \citet{Harris:10.1007/3-540-45414-4_21} and was used in many subsequent algorithms.
The algorithms we propose extend lock-free algorithms that employ this method.
In the description, we say ``mark a node'' to mean that a node is marked for removal in this manner.

\subsection{Recovery}\label{sub_sec.recovery}
The recovery procedure traverses all areas that contain persistent nodes.
It determines the nodes that currently belong to the set and reconstructs the linked data structure in the volatile 
memory to allow subsequent fast access to the nodes.
Note that this construction does not need to use \code{psync} operations.
Moreover, the reconstructed set may have a different structure from the one prior to the crash (for example, as a randomized skip list).
The sole purpose of the structure is to  make normal operations efficient.

The proposed algorithms require the recovery execution to complete before further operations can be applied.
Before completing recovery of the data structure on the volatile memory, the data structure is not coherent and cannot be used.
This is unlike some previous algorithms, such as \citep{Friedman:3178487.3178490, David:215967}, which allow the recovery and subsequent operations to run concurrently.
This requirement works well in a natural setting where crashes are infrequent.

\subsection{Link-Free Sets}\label{sub_sec.lf_overview}
The first algorithm we propose for implementing a durable lock-free set is called \emph{link-free}, as it does not persist links.
This algorithm keeps two validity bits in each node, allowing making a node as invalid while it is in a transient state before being inserted into the list.
A node is considered valid only if the value of both bits match.
Deciding if a node is in the set depends on whether it is valid and not logically deleted.
We follow \citet{Harris:10.1007/3-540-45414-4_21} and mark a node to make it logically deleted.
The complementary case is when the validity bits do not match, making the node invalid.
An invalid node is not in the set.

To determine whether a node is in the set or not, the {\em contains} operation checks that it is in the volatile set structure, i.e., that it is not marked as deleted.
If this is the case, the contains operation makes sure this node is valid and flushed so that this node will be resurrected if a crash and a recovery occur.
This ensures that the returned value of the contains matches the NVRAM view of the data structure's state. 

To insert a node, the node first needs to be initialized.
To this end, one validity bit is flipped, making the node invalid, and then the key and value are written into it. Intermediate states do not 
affect a future recovery because an invalid node is not recovered. 
Afterwards, the node is inserted into the linked structure and is made valid by flipping the second validity bit.
The insertion completes by executing a \code{psync} on the new node, making the node durably in the set.
If a node with the same key already exists, the previous insert is first helped by making the previously inserted node valid, and ensuring its content is flushed.
At this point, the {\em insert} can return and report failure due to the key already existing in the set.

To remove a node, the removal first helps complete the insertion of the target node. 
The node is made valid and then its \emph{next} pointer can be marked, so that it becomes logically deleted.
The removal is completed by executing a \code{psync} on the marked node.
If the node is already logically deleted, it is flushed using a \code{psync} and the thread returns reporting failure (as it was already deleted). 
During recovery, a marked node is considered not in the set.

Note that \code{psync} may be called multiple times on the same node.
To further reduce the number of \code{psync} operations, we employ an optimization.
Since the proposed algorithm persists a newly inserted node and a newly marked one, we use two flags to indicate whether a \code{psync} was executed
after inserting the node or after deleting it.
The first flag indicates that a new node was written to the NVRAM, and the second flag indicates that a deleted node was written back.
Before actually calling \code{psync} on the node, the insert (or remove, correspondingly) flag is checked to minimize the number of redundant \code{psync} operations.
After calling \code{psync} on a node, the insert (or remove, correspondingly) flag is set.
This way threads coming in a later point see that the flags are set, and they do not execute an unnecessary \code{psync}.
This is an extension of the \emph{link-and-persist} optimization presented by \citet{David:215967}.

\subsection{SOFT: Sets with an Optimal Flushing Technique} \label{sub_sec.soft_overview}
The second algorithm we introduce is \emph{\SOFT{}} (Sets with an Optimal Flushing Technique).
\textsc{Soft} is also a durable lock-free algorithm for a set.
It requires the minimal theoretical number of fences per operation. Specifically, each thread performs at most one fence per update and zero fences per read operation \citep{Cohen:3210377.3210400}.

Each key in the set has two separate representations in memory: the persistent node and the volatile node.
Similarly to our link-free algorithm, \emph{persistent nodes} (PNodes) are stored in the durable areas.
They contain a key and its associated value and three validity bits used for a similar but extended validity scheme.
Each time we wish to write to the NVRAM, we do so via a PNode method.
The PNode methods are described in further detail in Section~\ref{sub_sec.pnode}.

The volatile node takes part in the volatile-linked graph of the set.
In addition to holding the key and value, it has a pointer to a PNode with the same key and value, and pointers to its descendants in the linked structure.
The pointer, which is usually used for marking, is used to keep a state that indicates the condition the node is in.
A node can be in one of the following four states:
\begin{enumerate}
    \item Inserted: The node is in the set, is linked to the structure in the volatile memory and its PNode has been written to the NVRAM. \label{state.i}
    \item Deleted: The node is not in the set.
    In this case, the node can be unlinked from the volatile structure and later freed. \label{state.d}
    \item Intention to Insert: The node is in the middle of being inserted, and its PNode is not yet guaranteed to be written to the NVRAM. \label{state.iti}
    \item Inserted with Intention to Delete: The node is in the middle of being removed, and its removed condition is not yet guaranteed to be written to the NVRAM. \label{state.itd}
\end{enumerate}

The read operation (contains) executes on the volatile structure and does not require any \code{psync} operations, which is in line with the bound.
A contains operation only reads the state of the relevant node and acts accordingly.
A node that is either ``inserted'' or ``inserted with intention to delete'' is considered a part of the set, so contains returns \code{true}.
Nodes with one of the remaining states (``intention to insert'' or ``deleted'') cause the contains operation to return \code{false}.

To add a node to the set, \SOFT{} allocates a volatile node and a PNode, links them together, and fixes the volatile node's state to be ``intention to insert''.
Next, the insert operation adds the node to the volatile structure.
Read operations seeing the node in this state do not consider it as a part of the set. 
Thereafter, the associated PNode is written to the NVRAM and the state of the volatile node is changed to ``inserted''.
When the state is ``inserted'', other operations see the key of this node as a part of the set.

When trying to insert a node into the volatile structure, if there is a node with the same key in the set, the node's state is checked.
If the state of this node is ``inserted'' or ``inserted with intention to delete'', the node might be in the set in the event of a crash, so the thread fails right away.
If the state is ``intention to insert'', then the old node is not yet in the set, so the current thread helps complete the insertion before failing.
Just as many other algorithms, in \SOFT{}, deleted nodes are trimmed when traversing the linked-structure of the set, so there is no need to consider the scenario of seeing a node with the ``deleted'' state.
Either way, only a single \code{psync} operation is executed, following the theoretical bound.

When a remove operation wishes to remove a node, it must ensure the relevant node is in the set.
A remove operation changes the node's state from ``inserted'' to ``inserted with intention to delete''.
In this case, read operations do acknowledge the node because the removal has not finished yet.
Then the removal is written to the NVRAM and, finally, the state changes to ``deleted''.
A node with the state ``intention to insert'' cannot be removed because it is not yet in the set.
In this case, the remove operation can return a failure: there is no node in the set with the given key.
Alternatively, the state of the node the thread wishes to remove may already be ``inserted with intention to delete''.
In this case, before failing, the thread helps completing the removal and persisting it.
Just as before, this operation is done using only a single \code{psync}.

The goal of the states is to make threads help each other complete operations and reduce the number of \code{psync} operations to the minimum. 
States~\ref{state.iti} and \ref{state.itd}, described above, are used as flags to indicate the beginning of an operation so other threads are able to help.

Both insert and remove use the same logic.
They first update the non-volatile memory, and only then execute the operation (reaching a linearization point) on the volatile structure.
In other words, the state a thread sees in \SOFT{} already resides in the NVRAM, unlike link-free in which a node has to be written back to the NVRAM.
This logic follows the upper bound of \citet{Cohen:3210377.3210400}.

\section{The Details of the Link-Free Algorithm}\label{sec.lf_details}
In this section we described the \emph{link-free} linked list.
A link-free hash table is constructed simply as a table of buckets, where each bucket uses the link-free list to hold its items.
Extending this algorithm to a skip list is straightforward.

The link-free linked list uses a node to store an item in the set; see Listing~\ref{algo.lf_node}.
Unlike \SOFT, each key has a single representation in both volatile and non-volatile space.
Each node has two validity bits, two flags to reduce the number of \code{psync} operations, a key, a value and a \emph{next} pointer that also contains a marking bit to indicate a logical deletion \citep{Harris:10.1007/3-540-45414-4_21}.
\begin{lstlisting}[caption={Node Structure}, label=algo.lf_node]
class Node{
    atomic<byte> validityBits;
    atomic<bool> insertFlushFlag, deleteFlushFlag;
    long key;
    long value;
    atomic<Node*> next;
} aligned(cache line size);
\end{lstlisting}

Building on the implementation of \citet{Harris:10.1007/3-540-45414-4_21}, the list is initialized with a head sentinel node with key $-\infty$, and a tail sentinel node with key $\infty$. All the other nodes are inserted between these two, and are sorted in an ascending order.

\subsection{Auxiliary Functions}\label{sub_sec.lf_help}

Before explaining each operation, we first discuss the auxiliary functions.
We use the functions \code{isMarked}, \code{getRef}, and \code{mark} without providing their implementations since these are only bit operations, to clean, mark, or test the least significant bit of a pointer.
In addition, we use \code{FLUSH\_DELETE} and \code{FLUSH\_INSERT} to execute a \code{psync} operation to write the content of a node to the NVRAM when removing or inserting it from or into the list.
Before executing the \code{psync}, the appropriate (insert or delete) flag is used to check whether the latest modification to this node has already been written to the NVRAM so avoid repeated flushing.
Next, \code{flipV1} and \code{makeValid} modify the validity of a node: \code{flipV1} flips the value of the first validity bit, making the node invalid, and \code{makeValid} makes the node valid by equating the value of the second bit to the value of the first bit.

The auxiliary function \code{trim} (Listing~\ref{algo:lf_help}) unlinks \code{curr} from the list.
Just prior to the unlinking CAS (line~\ref{code.trim_unlink_cas}), node \code{curr} is flushed to make the delete mark on it persistent (line~\ref{code.trim_flush}).
The return value signifies whether the unlinking succeeded or not.

The \code{find} function (Listing~\ref{algo:lf_help}) traverses the list in order to locate nodes \code{curr} and \code{pred}.
The key of \code{curr} is greater or equal to the given key, and \code{pred} is the predecessor of \code{curr} in the list.
During its search of the list, \code{find} invokes \code{trim} on any marked (logically deleted) node (line~\ref{code.call_trim}).

\begin{lstlisting}[caption={Auxiliary functions}, label=algo:lf_help]
bool trim(Node *pred, Node *curr){
	FLUSH_DELETE(curr); @\label{code.trim_flush}@
	Node *succ = getRef(curr->next.load());
	return pred->next.compare_exchange_strong(curr, succ); @\label{code.trim_unlink_cas}@
}

Node*, Node* find(long key){//method returns two pointers, pred and curr.
    Node* pred = head, *curr = head->next.load();
    while(true){
        if(!isMarked(curr->next.load())){
            if(curr->key >= key) @\label{code.lf_if_bigger}@
                break;
            pred = curr;
        }
        else
            trim(pred, curr); @\label{code.call_trim}@
        curr = getRef(curr->next.load()); @\label{code.lf_find_read}@
    }
    return pred, curr;
}
\end{lstlisting}

\subsection{The contains Operation}\label{sub_sec.lf_contains}

The contains operation, based on the optimization of \citet{heller2005lazy}, is wait-free unlike the lock-free insert and remove operations.
Given a key, it returns \code{true} if a node with that key is in the list and \code{false} otherwise.

In lines~\ref{code.lf_start_loop} -- \ref{code.lf_end_loop} (Listing~\ref{algo:lf_contains}), the list is traversed in order to find the requested key.
If a node with the given key is not found, then the operation returns \code{false} (line~\ref{code.lf_key_not_found}).
If the node exists but has been marked, it is flushed and the thread returns \code{false} (line~\ref{code.lf_node_deleted}).
The last possible case is that the node exists and has not been marked as removed.
In this case, the node is made valid, is flushed to make its insertion visible after a crash, and \code{true} is returned (line~\ref{code.lf_node_exists}).

\begin{lstlisting}[caption={Link-Free List contains}, label=algo:lf_contains]
bool contains(long key){
    Node* curr = head->next.load();
    while(curr->key < key) @\label{code.lf_start_loop}@
        curr = getRef(curr->next.load()); @\label{code.lf_end_loop}@
    if(curr->key != key) @\label{code.lf_key_not_found}@
        return false; @\label{code.lf_key_not_found_return_false}@
    if(isMarked(curr->next.load())){@\label{code.lf_node_deleted}@
        FLUSH_DELETE(curr); @\label{code.contains_flush_marked}@
        return false; @\label{code.contains_flush_marked_return_false}@
    }
    makeValid(curr);  @\label{code.lf_node_exists}@
    FLUSH_INSERT(curr); @\label{code.lf_contains_line_after_makevalid}@
    return true;
}
\end{lstlisting}

\subsection{The insert Operation}\label{sub_sec.lf_insert}
The insert operation adds a key-value pair to the list.
It returns \code{true} if the insertion succeeds (i.e., the key was not in the list) and \code{false} otherwise.

The insert initiates a call to find, in order to know where to link the newly created node (line~\ref{code.insert_find}).
If the key does not exist, the operation allocates a new node out of a durable area using \code{allocFromArea()}.
The allocation procedure (Section~\ref{sec.mm}) returns a node that is available for use and whose validity state is valid, i.e., both validity bits have the same value.
The insert operation then makes the node invalid by changing the first validity bit (line~\ref{code.make_node_invalid} Listing~\ref{algo:lf_insert}).
The subsequent memory fence keeps the order between the writes and guarantees that the node becomes invalid before its initialization.
This ensures that an incomplete node initialization will not confuse the recovery.
Next, the operation initializes the node's fields, including the \emph{next} pointer of the node (line~\ref{code.init_next}), and then the operation tries to link the new node using a CAS (line~\ref{code.insert_cas}).
Note that the node is still invalid when linking it to the list.
If the CAS fails, the entire operation is restarted and, if successful, the new node is made valid by flipping the second validity bit (line~\ref{code.insert_make_valid}).
It is then flushed to persist the insertion and \code{true} is returned.

If the key exists in the list, the existing node is made valid, then flushed and the operation returns \code{false} (lines~\ref{code.lf_insert_node_exists} -- \ref{code.lf_insert_node_exists_returning_false}).
When finding a node with the same key, the existing node might not be valid yet because the node is linked to the list in an invalid state.
It has to be made valid and persistent before \code{false} can be returned.
Otherwise, a subsequent crash may reflect this failed insert but not reflect the preceding insert that caused this failure.
This ensures durable linearizability.

The order between making the node valid and linking it is important.
Making a node valid first and then linking it may cause inconsistencies.
Consider a scenario with two threads trying to insert a node with a key $k$ but with different values.
Both threads may finish initializing their nodes and make them valid, but then the system crashes.
During recovery, both nodes are found in a valid state (they may appear in the NVRAM even if an explicit flush was not executed), and there is no way to determine which should be in the set and which should not.

\begin{lstlisting}[caption={Link-Free List insert}, label=algo:lf_insert]
bool insert(long key, long value){
    while(true){
        Node *pred, *curr;
        pred, curr = find(key); @\label{code.insert_find}@
        if(curr->key == key){
            makeValid(curr); @\label{code.lf_insert_node_exists}@
            FLUSH_INSERT(curr); @\label{code.lf_insert_line_after_makevalid}@
            return false; @\label{code.lf_insert_node_exists_returning_false}@
        }
        
        Node* newNode = allocFromArea(); @\label{code.link_free_alloc_node}@
        flipV1(newNode); @\label{code.make_node_invalid}@
        atomic_thread_fence(memory_order_release);
        newNode->key = key; @\label{code.init_node_fields}@
        newNode->value = value;
        newNode->next.store(curr, memory_order_relaxed); @\label{code.init_next}@
        if(pred->next.compare_exchange_strong(curr, newNode)){ @\label{code.insert_cas}@
            makeValid(newNode); @\label{code.insert_make_valid}@
            FLUSH_INSERT(newNode); @\label{code.insert_flush_newnode}@
            return true; @\label{code.insert_last_line}@
        }
    }
}
\end{lstlisting}

\subsection{The remove Operation}\label{sub_sec.lf_remove}
Given a key, the remove operation deletes the node with that key from the set.
The return value is \code{true} when the removal was successful, i.e., there was such a node in the list, and now there is not, and \code{false} otherwise.

First, the requested node and its predecessor are found (line~\ref{code.remove_find} Listing~\ref{algo:lf_remove}).
If the node found does not contain the given key, the thread returns \code{false}.
Otherwise, the node is made valid and then its \emph{next} pointer is marked using a CAS (line~\ref{code.logically_remove}).
All along the code (and also here) we maintain the invariant that a marked node is valid.
If the CAS succeeds, the operation finishes by calling \code{trim} to physically remove the node, and otherwise the removal is restarted.

There is no need for a \code{psync} operation between making \code{curr} valid (line~\ref{code.remove_make_valid}) and the logical removal (line~\ref{code.logically_remove}).
Both modify the same cache line and the writes to the cache are ordered by the CAS (with default \code{memory\_order\_seq\_cst}), implying the same order to the NVRAM.
Therefore, the view of the node can be invalid and not marked (prior to line~\ref{code.remove_make_valid}), valid and not marked (between lines~\ref{code.remove_make_valid} and \ref{code.logically_remove}), or valid and marked (after line~\ref{code.logically_remove}). 
The node can never be in an inconsistent state (marked and invalid).

\begin{lstlisting}[caption={Link-Free List remove}, label=algo:lf_remove]
bool remove(long key){
    bool result = false;
    while(!result){
        Node *pred, *curr;
        pred, curr = find(key); @\label{code.remove_find}@
        if(curr->key != key) @\label{code.remove_node_does_not_exist}@
            return false;
        Node* succ = getRef(curr->next.load());
        Node* markedSucc = mark(succ);
        makeValid(curr); @\label{code.remove_make_valid}@
        result = curr->next.compare_exchange_strong(succ, markedSucc); @\label{code.logically_remove}@
    }
    trim(pred, curr);
    return true;
}
\end{lstlisting}

\subsection{Recovery}\label{sub_sec:lf_recovery}
The validity scheme we use helps us determine whether a node was linked to the list before a crash occurred.
This is possible because before initializing a node, it is made invalid so no partial writes are observed.
If a remove operation manages to mark a node, we can know for sure it is removed.

The recovery takes place after a crash and the data it sees is data that was flushed to the NVRAM 
prior to the crash.
The procedure starts by initializing an empty list with a head and a tail.
Afterwards, it scans the durable areas of the threads for nodes.
All nodes that are valid and unmarked are inserted, one by one, to an initially empty link-free list. All other nodes (invalid nodes and valid and marked nodes) are sent to the memory manager for reclamation.
The linking of the valid nodes is done without any \code{psync} operations since all the data in the nodes is already stored in the NVRAM.
\section{The Details of SOFT}\label{sec.soft_details}
The second algorithm we present is \SOFT{}, whose number of \code{psync} operations matches the theoretical lower bound (of \citep{Cohen:3210377.3210400}).
It does so by dividing each update operation into two stages: intention and completion.
In the intention stage, a thread triggers helping mechanisms by other threads, while not changing the logical state of the data structure.
After the intention is declared, the operation becomes durable, in the sense that a subsequent crash will trigger a recovery that will attempt to execute the operation.
In this section, we start by describing the nodes of the \SOFT{} list (Sections~\ref{sub_sec.pnode} and \ref{sub_sec.vnode}), then we discuss the implementation details of each set operation (Sections~\ref{sub_sec.soft_conatins}, \ref{sub_sec.soft_insert} and \ref{sub_sec.soft_remove}), and finally in Section~\ref{sub_sec:soft_recovery}, we explain the recovery scheme.

\subsection{PNode}\label{sub_sec.pnode}
At the core of \SOFT{} there is a \emph{persistent node} (PNode) that captures the state of a given key in the NVRAM.
It has a key, a value and three flags, which are described next. The structure of the persistent node is provided in Listing~\ref{algo.soft_pnode}. 

\begin{lstlisting}[caption={PNode}, label=algo.soft_pnode]
class PNode{
    atomic<bool> validStart, validEnd, deleted;
    atomic<long> key;
    atomic<long> value;
} aligned(cache line size);
\end{lstlisting}

The PNode's three flags indicate the state of the node in the NVRAM. 
The first two flags have a similar meaning to the ones used by the link-free algorithm.
When both flags are equal, the node is in a consistent state, and if the flags are different, then the node is in the middle of being inserted.
It also has an additional flag indicating whether the node was removed. 

Specifically, the PNode starts off with all three flags having the same value, {\em pInitialValidity}. 
In this case, the PNode is considered {\em valid} and {\em removed}. 
The negation of pInitialValidity is returned to the user of the node after calling \code{alloc}, and is denoted {\em pValidity}. 
From this point on, the state of the persistent node progresses by flipping the flags from pInitialValidity to pValidity.

When a key-value pair is inserted into the data structure, the corresponding PNode is made valid, by setting {\em validStart} to pValidity, assigning the key and the value of the node, and finally setting {\em validEnd} to pValidity.
Only then, the persistent node is written to the NVRAM. 
When validStart differs from validEnd, the node is considered {\em invalid}. 
When validStart equals to validEnd (but is still different from deleted), the node is properly inserted and will be considered during recovery. 

When the PNode is removed from the data structure, the {\em deleted} flag is set and the node is flushed. 
Then, the node is {\em valid} and {\em removed}, so it is not considered during recovery. 
Note that this represents exactly the same state as when the node was allocated, making the persistent node ready for future allocations. 
The only difference is the value of all flags, which was swapped from pInitialValidity to pValidity.
Code for allocating, creating and destroying a PNode appears in Listing~\ref{algo:soft_pnode_funcs}.

\begin{lstlisting}[caption={PNode Member Functions}, label=algo:soft_pnode_funcs]
bool PNode::alloc(){
    return !validStart.load();
}

void PNode::create(long key, long value, bool pValidity){ @\label{code.start_pnode_create}@
    validStart.store(pValidity, memory_order_relaxed); @\label{code.flip_v1}@
    atomic_thread_fence(memory_order_release);
    this->key.store(key, memory_order_relaxed);
    this->value.store(value, memory_order_relaxed);
    validEnd.store(pValidity, memory_order_release); @\label{code.flip_v2}@
    psync(this);
} @\label{code.end_pnode_create}@

void PNode::destroy(bool pValidity){ @\label{code.start_pnode_destroy}@
    deleted.store(pValidity, memory_order_release); @\label{code.destroy_pnode}@
    psync(this);
} @\label{code.end_pnode_destroy}@
\end{lstlisting}

\subsection{Volatile Node}\label{sub_sec.vnode}
Volatile nodes have a key, a value, and a \emph{next} pointer (to the next volatile node). 
In addition, they contain a pointer to a persistent node (i.e., a PNode, explained in Section~\ref{sub_sec.pnode}) and {\em pValidity}, a Boolean flag indicating the pValidity of the persistent node. 
The structure of the volatile node appears in Listing~\ref{algo.soft_node}.

\begin{lstlisting}[caption={Volatile Node}, label=algo.soft_node]
class Node{
	long key;
	long value;
	PNode* pptr;
	bool pValidity;
	atomic<Node*> next;
};
\end{lstlisting}

Similar to the lock-free linked list algorithm by \citet{Harris:10.1007/3-540-45414-4_21}, the last bits of the \emph{next} pointers store whether the node is deleted. 
Unlike Harris' algorithm, a volatile node must be in one of four states: ``intention to insert'', ``inserted'', ``inserted with intention to delete'', and ``deleted'', as discussed in the overview (Section~\ref{sub_sec.soft_overview}). 
We assume standard methods for handling pointers with embedded state (lines~\ref{line:stateMethods} -- \ref{code:last_psuedo} Listing~\ref{algo:soft_contains}). 
In addition, we use \code{trim} and \code{find} to physically unlink removed nodes and find the relevant window, respectively (Listing~\ref{algo:soft_find}).
Unlike its link-free counterpart, find also returns the state of both nodes. 
One is in the second address returned and the other is returned explicitly.
Moreover, trim does not execute a \code{psync} before unlinking a node.

\begin{lstlisting}[caption={find and trim}, label=algo:soft_find]
bool trim(Node *pred, Node *curr) {
	state predState = getState(curr);
	Node *currRef = getRef(curr), *succ = getRef(currRef->next.load());
	succ = createRef(succ, predState);
	return pred->next.compare_exchange_strong(curr, succ); @\label{code.soft_trim_cas}@
}

Node*, Node* find(long key, state *currStatePtr){
	Node *pred = head, *curr = pred->next.load();
	Node *currRef = getRef(curr);
	state predState = getState(curr), cState;
	while (true){
		Node *succ = currRef->next.load(); @\label{code.soft_find_read_next}@
		Node *succRef = getRef(succ);
		cState = getState(succ);
		if (cState != DELETED){ @\label{code.soft_trim_deleted}@
			if (currRef->key >= key) @\label{code.soft_find_check_key}@
				break;
			pred = currRef;
			predState = cState;
		}
		else
			trim(pred, curr); @\label{code.soft_find_trim}@
		curr = createRef(succRef, predState);
		currRef = succRef;
	}
	*currStatePtr = cState;
	return pred, curr;
}
\end{lstlisting}

\subsection{The contains Operation} \label{sub_sec.soft_conatins}

The contains operation checks whether a key resides in the set. 
Unlike the insert and remove operations, contains is wait-free and does not use any \code{psync} operations.

A node is in the set only if its state is either ``inserted'' or ``inserted with intention to delete''.
A node with the state ``inserted with intention to delete'' is still in the set because there is a thread trying to remove it, but it has not finished yet.
Only in these two cases the return value is \code{true}; in all the other cases, it is \code{false}.

\begin{lstlisting}[caption={SOFT List contains}, label=algo:soft_contains]
//Pseudo-code for managing state pointers
#def createRef(address, state) {.ptr=address, .state=state} @\label{line:stateMethods}@
#def getRef(sPointer) {sPointer.ptr}
#def getState(sPointer) {sPointer.state}
#def stateCAS(sPointer, oldState, newState) {old=sPointer.load();
    return sPointer.compare_exchange_strong(createRef(old.ptr, oldState),
    createRef(old.ptr, newState));} @\label{code:last_psuedo}@

bool contains(long key){
	Node *curr = head->next.load(); @\label{code.first_line}@
	while (curr->key < key) @\label{code.soft_contains_loop_start}@
		curr = getRef(curr->next.load()); @\label{code.soft_contains_loop_end}@
	state currState = getState(curr->next.load()); @\label{code.read_state}@
	if(curr->key != key)
        return false;
    if(currState == DELETED || currState == INTEND_TO_INSERT)
        return false;
    return true;
	@\ignore{return (curr->key == key) && (currState == INSERTED || currState == INTEND_TO_DELETE); @\label{code.soft_contains_ret}}@
}
\end{lstlisting}

\subsection{The insert Operation} \label{sub_sec.soft_insert}
Insertion in \SOFT{} follows the standard set API, which is getting a key and a value and inserting them into the set. 
The operation returns whether the insertion was successful. 
Code is provided in Listing~\ref{algo:soft_insert} and is discussed below. 

Similar to link-free, persistent nodes are allocated from a durable area using \code{allocFromArea}.
When allocating a new PNode, all its validity bits have the same value, so its state is deleted.
Volatile nodes can be allocated from the main heap.

The first step of insert is a call to \code{find}, which returns the relevant window (line~\ref{code.soft_insert_find}).
As mentioned above, while traversing the list, if a logically removed node, is found along the way the thread tries to complete its physical removal.
Unlike link-free, however, there is no need to execute a \code{psync} a removed node before unlinking it. 
The volatile node becomes removed only after the corresponding PNode becomes removed and is written to the NVRAM.
Therefore, if the state of a volatile node is ``deleted'', it is always safe to unlink it from the list and it does not require further operations.

Discovering a node with the same key already in the list fails the insertion.
Nonetheless, the thread needs to help complete the insertion operation before returning, if the found node's state is ``intention to insert''. 
In the complementary case, when there is no node with the same key, the thread allocates a new PNode and a new volatile node, and attempts to link the latter node to the list (line~\ref{code.soft_link_node}) using a CAS.
The new volatile node is initialized with the state ``intention to insert'', because we want other threads to help with finishing the insertion.
If the CAS failed, the entire operation starts over.
Otherwise, the thread moves to the helping part (lines~\ref{code.soft_init_pnode} -- \ref{code.soft_complete_insert}), where the node is fully inserted.

The helping part starts by initializing the PNode of the appropriate node (line~\ref{code.soft_init_pnode}).
Afterwards, all the threads try to complete the insertion and make it visible by changing the state of the new node to ``inserted'' (line~\ref{code.soft_complete_insert}).
Finally, the thread returns \code{true} or \code{false} depending on the path taken.

\begin{lstlisting}[caption={SOFT List insert}, label=algo:soft_insert]
bool insert(long key, long value){
    Node *pred, *curr, *currRef, *resultNode;
    state predState, currState;
    
    while(true){ @\label{code.soft_insert_start_link_loop}@
        pred, curr = find(key, &currState); @\label{code.soft_insert_find}@
        currRef = getRef(curr);
        predState = getState(curr);
        bool result = false;
        if(currRef->key == key){ @\label{code.soft_node_exists}@
            if(currState != INTEND_TO_INSERT) @\label{code.soft_fail_insert}@
                return false;
            resultNode = currRef;
            break;
        }
        else{
            PNode* newPNode = allocFromArea(); @\label{code.allocate_new_pnode}@
            Node* newNode = new Node(key, value, newPNode, newPNode->alloc()); @\label{code.allocate_new_node}@
            newNode->next.store(createRef(currRef, INTEND_TO_INSERT),
                memory_order_relaxed); @\label{code.soft_init_state}@
            
            if(!pred->next.compare_exchange_strong(curr,
                createRef(newNode, predState))) @\label{code.soft_link_node}@
                continue;
            resultNode = newNode;
            result = true; @\label{code.soft_insert_result_true}@
            break;
        }
    }@\label{code.soft_insert_end_link_loop}@    
    resultNode->pptr->create(resultNode->key, resultNode->value, @\label{code.soft_init_pnode}@
        resultNode->pValidity);
    while(getState(resultNode->next.load()) == INTEND_TO_INSERT) @\label{code.soft_insert_check_state}@
        stateCAS(&resultNode->next, INTEND_TO_INSERT, INSERTED); @\label{code.soft_complete_insert}@
            
    return result;
}
\end{lstlisting}

\subsection{The remove Operation} \label{sub_sec.soft_remove}
The remove operation unlinks a node from the set with the same key as the given key.
It returns \code{true} when the removal succeeds and \code{false} otherwise.

Similar to the previous operation, remove starts by finding the required window.
If the key is not found in the set, the operation returns \code{false}. 
Recall that a volatile node is removed from the set only after its PNode becomes deleted in the NVRAM, so returning \code{false} is safe. 
Also, if the found node has a state of ``intention to insert'', the remove operation returns \code{false}.
This is because such a node is not guaranteed to have a valid PNode in the NVRAM.

In the case when a node with the correct key is found, the thread attempts to mark the node as ``inserted with intention to delete''. 
At this point, all threads attempting to remove the node compete; the successful thread will return \code{true} while other threads will return \code{false} (line~\ref{code.soft_mark_node}). 
This does not, however, change the logical status of the node (the key is still considered as inserted) or modify the NVRAM.
Once the node is made ``inserted with intention to delete'', the thread calls \code{destroy} on the relevant PNode, so that the deletion is written to the NVRAM. 
Finally, the state is changed to be ``deleted'' to indicate the completion and the result is returned.
Note that calling \code{destroy} and marking the node as ``deleted'' happens even if the thread fails in the ``inserted with intention to delete'' competition, in which case it helps the winning thread.
The final step, executed only by the thread that won the ``inserted with intention to delete'' competition, physically disconnects the node from the list by calling \code{trim}.
This latter step does not change the logical representation of the set and is executed only by a single thread to reduce contention.

\begin{lstlisting}[caption={SOFT List remove}, label=algo:soft_remove]
bool remove(long key){
    bool result = false;
    Node *pred, *curr;
    state predState, currState;
    pred, curr = find(key, &currState); @\label{code.soft_remove_find}@
    Node* currRef = getRef(curr);
    predState = getState(curr);
    if(currRef->key != key) @\label{code.soft_fail_remove_1}@
        return false; @\label{code.soft_fail_remove_1_return}@
    if(currState == INTEND_TO_INSERT) @\label{code.soft_fail_remove_2}@
        return false; @\label{code.soft_fail_remove_2_return}@
    
    while(!result && getState(currRef->next.load()) == INSERTED) @\label{code.soft_while1}@
        result = stateCAS(&currRef->next, INSERTED, INTEND_TO_DELETE); @\label{code.soft_mark_node}@
    currRef->pptr->destroy(currRef->pValidity); @\label{code.flush_delete}@
    while(getState(currRef->next.load()) == INTEND_TO_DELETE) @\label{code.soft_while2}@
        stateCAS(&currRef->next, INTEND_TO_DELETE, DELETED); @\label{code.soft_complete_remove}@
    
    if(result)
        trim(pred, curr); @\label{code.soft_remove_trim}@
    return result; @\label{code.soft_return_result}@
}
\end{lstlisting}
\subsection{Recovery}\label{sub_sec:soft_recovery}
In \SOFT{} only the PNodes are allocated from the durable areas.
All the volatile nodes are lost due to the crash.
This means that the intentions are not available to the recovery procedure, so it decides whether a key is a part of the list based on the validity bits kept in the PNode.
A PNode is valid and a part of the set, if the first two flags (\code{validStart} and \code{validEnd}) have the same value, and the last flag (\code{deleted}) has a different value.

In order to reconstruct the \SOFT{} list, a new and empty list is allocated.
Then the recovery iterates over the durable areas to find valid and not deleted PNodes.
If such a PNode $pn$ is found, a new volatile node $n$ is allocated and its fields are initialized using the $pn$'s data.
The value of $n$'s \code{pValidity} is set to the be $pn$'s \code{validStart}, and \code{pptr} points to $pn$.
Finally, $n$ is linked to the list in a sorted manner and its state is set to ``inserted''.
Similar to link-free, no \code{psync} operations are used to link $n$ since the data in $pn$ already persisted in the NVRAM.
Invalid or deleted PNodes are sent to the memory manager for reclamation.
\section{Memory Management}\label{sec.mm}
Both of our algorithms use durable areas in which we keep the nodes with persistent data, which are used by the recovery procedure.
A memory manager allocates new nodes and new areas, keeps record of old ones, and has free-lists for each thread.
Moreover, since this is a lock-free environment, our algorithms are susceptible to the ABA problem \citep{Michael:564870.564881} and to use-after-free.

To maintain the lock-freedom of our algorithms, lock-free memory reclamation schemes can be used (e.g., \citep{cohen2018every, michael2004hazard, Alistarh:2017:FCM:3064176.3064214, Balmau:2016:FRM:2935764.2935790, Dice:2016:FNM:2926697.2926699, cohen2015efficient, Brown:2015:RML:2767386.2767436}).
Some, however, are complicated to incorporate; some require the data structure to be in a normalized form; and others have significant overhead that commonly deteriorates performance.
We, therefore, chose to employ the very simple \emph{Epoch Based Reclamation} scheme (EBR) \citep{fraser2004practical} that is not lock-free but it performs very well and provides progress for the memory management when the threads are not stuck.

In EBR we have a global counter to indicate the current epoch, and each thread is either in an epoch (when executing a data structure operation) or \emph{idle}.
A thread joins the current epoch at the beginning of each operation, and becomes idle at its end.
When an object is freed, it is added to a free-list for the current epoch.
Whenever a thread runs out of memory, it starts the reclamation of the current epoch, denoted $e$.
When all the threads reach either epoch $e$ or an idle state, all the objects in the free-list related to epoch $e-2$ can safely be reclaimed and reused.
We used a variant of EBR that uses clock vectors. In particular, we used \code{ssmem}, an EBR that accompanies the ASCYLIB algorithms \citep{David:2694344.2694359}.

The \code{ssmem} allocator normally serves volatile memory, allocating objects of fixed predetermined size.
We adapted it to our setting.
In \code{ssmem}, each thread has its own personal allocator so the communication between different threads is minimal.
The allocator provides an interface that allows allocating and freeing of objects of a fixed size in specially allocated designated areas.
It initially allocates a big chuck of memory from which it returns objects to the program using a bump pointer.
When the area fills up, nodes get reclaimed, and holes emerge; a \emph{free-list} is then used to allocate objects.
Each thread has it own free-list so freeing nodes or using free ones does not require any form of synchronization.
The free-lists are volatile and are reconstructed during a recovery.
Invalid or deleted nodes a thread encounters during recovery while traversing the durable areas are inserted into the private free-list of the thread.

The memory manager keeps a list of all the areas it allocated so it can free them at the end of the execution.
Throughout its life, the original \code{ssmem} manager does not free areas back to the operating system.
In our implementation, empty areas can be returned to the operating system during the recovery if all the nodes of an area are free.

Both link-free and \SOFT{} use durable areas as a part of their memory allocation scheme.
These are address spaces in the heap memory that are used solely for node allocation and, therefore, \code{ssmem} can be used with small modifications.
When a thread performs an insertion, it allocates a node from these areas, and when a node is removed, it is returned to the proper free-list.
To reduce false sharing and contention, each thread has its own areas.

Using \code{ssmem}, each thread keeps a private list with one node per allocated area pointing to all the areas it allocated throughout the execution, denoted \emph{area list}.
This list has to be persistent so after a crash the areas will not be lost.
We call nodes is this list \emph{area nodes}.
When allocating an additional area, we write its address in a new area node and write the new area node to the NVRAM.
Then, we link it to the beginning of the area list (there is no need for any synchronization since the area list is thread-local), and flush the link to it, making the new area node persistent.
The area list is persistent and its head is kept in a persistent thread-local space, which a recovery procedure can access.
Thus, all the addresses of the different areas can be traced after a crash and all persistent nodes can be traversed.

There is an inherent problem when using durable algorithms without proper memory management.
When inserting a new node, the node is allocated and only afterwards linked to the set.
In the case of deletion, the node is unlinked from the set, and subsequently can be freed.
Since a crash may occur at any time, we might have a persistent memory leak if a new node was not linked or if a deleted node was not freed.

Typically, this problem is solved by using a logging mechanism that records the intention (inserting or removing) along with the relevant addresses.
This way, in case of a crash, the memory leaks may be fixed by reading the records.
This logging mechanism requires more writes to the NVRAM, which take time, resulting in increased operation latency and worse throughput.

The durable areas solve this problem in a simpler manner since all the memory is allocated only from them.
Therefore, when recovering and traversing the different areas, leaks will be identified using the validity scheme.
Removed or invalid nodes can be freed and reused.
\section{Evaluation}\label{sec.results}

We ran the measurements of the link-free and \SOFT{} algorithms and compared them to the state-of-the-art set algorithm proposed by \citet{David:215967}.
We ran the experiments on a machine with 64 cores, with 4 AMD Opteron(TM) 6376 2.3GHz processors (16 cores each).
The machine has 128GB RAM, 16KB L1 per one core, 2MB L2 for every pair of cores and 6MB LLC per 8 cores (half a processor).
The machine's operating system is Ubuntu 16.04.6 LTS (kernel version 5.0.0). 
All the code was written in C++ 11 and compiled using g++ version 8.3.0 with a -O3 optimization flag.

NVRAM is yet to be commercially available, so following previous work \citep{volos2011mnemosyne, Wang:2014:SLT:2732951.2732960, Chakrabarti:2014:ALL:2714064.2660224, Arulraj:2015:LTS:2723372.2749441, schwalb2015nvc, kolli2016high, Cohen:3152284.3133891, David:215967,  Friedman:3178487.3178490, Cohen:2019:FCI:3297858.3304046, Ben-David:2019:DCF:3323165.3323187}, we measured the performance using a DRAM.
NVRAM is expected to be somewhat slower than DRAM \citep{Arulraj:2015:LTS:2723372.2749441, volos2011mnemosyne, Wang:2014:SLT:2732951.2732960}. Nevertheless, we assume that data becomes durable once it reaches the memory controller\footnote{\href{https://software.intel.com/en-us/blogs/2016/09/12/deprecate-pcommit-instruction}{https://software.intel.com/en-us/blogs/2016/09/12/deprecate-pcommit-instruction}}.
Therefore, we do not introduce additional latencies to NVRAM accesses.

Link-free and \SOFT{} use the \code{clflush} instruction to ensure that data is written back to the NVRAM (or to the memory controller). 
This instruction is ordered with respect to store operations \citep{Intel}, so an additional store fence is not required (unlike the \code{clflushopt} instruction, which does require a fence). 
\citet{David:215967} used a simulation of \code{clwb} (an instruction that forces a write back without invalidating the cache line, which is not supported by all systems). To compare apples to apples, we changed the code to execute a \code{clflush} instead (as other measured algorithms).

\subsection{Throughput Measurements}\label{sub_sec.mops}

We compared the algorithms to each other on three different fronts.
Each test consisted of ten iterations, five seconds each and the results shown in the graphs, are the average of these iterations.
In each test, the set was filled with half of the key range, aiming at a 50-50 chance of success for the insert and remove operations.
Error bars represent 99\% confidence intervals. 

First, we measured the scalability of each algorithm, i.e., the outcome of adding more threads to increase the parallelism. 
The workload was fixed to 90\% read operations (a common practice when evaluating sets \citep{herlihy2011art}), and the key range was fixed as well.
When running the lists, the key ranges were 256 and 1024.
We chose to run two tests with the lists so we could have a closer look at the effect of a longer list on the scalability and performance.
We also evaluated the hash set.
For the hash set, we used a larger key range of 1M keys with a load factor of 1.

\begin{figure}
\begin{subfigure}[t]{\textwidth}
    \begin{subfigure}[t]{0.515\linewidth}
\begin{tikzpicture}
    \begin{axis}[
    xlabel={\#Threads},
    ylabel={MOPS},
    xmin=0, xmax=70,
    ymin=0, ymax=25,
    xtick={0,10,20,30,40,50,60,70},
    ytick={0,5,10,15,20,25},
    legend pos=south east,
]
 
\addplot[
    color=LogFree,
    mark=x,
    error bars/.cd, y dir=both,y explicit
    ]
    coordinates {
(1,3.953252)+-(0.0812773425635537,0.0812773425635537)
(2,6.910791)+-(0.0811533760908366,0.0811533760908366)
(4,10.40007)+-(0.0988109163684742,0.0988109163684742)
(8,15.2716444444444)+-(0.0641517218433672,0.0641517218433672)
(16,18.0444333333333)+-(0.110901753502183,0.110901753502183)
(32,14.6330625)+-(0.0693935568241498,0.0693935568241498)
(64,9.17980333333333)+-(0.78495581125241,0.78495581125241)
    };
    
\addplot[
    color=LinkFree,
    mark=o,
    error bars/.cd, y dir=both,y explicit
]
coordinates {
(1,5.33942)+-(0.00242181442680082,0.00242181442680082)
(2,10.82969)+-(0.0271304265285596,0.0271304265285596)
(4,15.25731)+-(0.0922651967322745,0.0922651967322745)
(8,20.20735)+-(0.192814679294223,0.192814679294223)
(16,21.71689)+-(0.0874004579296794,0.0874004579296794)
(32,19.77761)+-(0.0276369996470444,0.0276369996470444)
(64,17.8634)+-(0.0345419113730217,0.0345419113730217)
};

\addplot[
    color=SOFT,
    mark=square,
    error bars/.cd, y dir=both,y explicit
]
coordinates {
(1,5.563159)+-(0.00175328184050365,0.00175328184050365)
(2,11.13968)+-(0.00497407710737348,0.00497407710737348)
(4,15.94176)+-(0.180510838293969,0.180510838293969)
(8,21.63993)+-(0.184328483512901,0.184328483512901)
(16,23.47221)+-(0.0596702501536291,0.0596702501536291)
(32,20.65547)+-(0.0386901546563054,0.0386901546563054)
(64,17.83098)+-(0.0336816289107587,0.0336816289107587)
};
 
\end{axis}
\end{tikzpicture}
\end{subfigure}
    \begin{subfigure}[t]{0.45\linewidth}
\begin{tikzpicture}
\begin{axis}[
    ybar,
    legend style={at={(0.5,-0.3)},
      anchor=north,legend columns=-1},
    ylabel={Relative Improvement},
    xlabel={\#Threads},
    symbolic x coords={1,2,4,8,16,32,64},
    xtick=data,
    bar width=\barwidth,
    nodes near coords align={vertical},
    xtick align=inside,
    ]
\addplot[color = LogFree,
        pattern=north east lines
    ] coordinates{
(1,1)
(2,1)
(4,1)
(8,1)
(16,1)
(32,1)
(64,1)
};
\addplot[color=LinkFree] coordinates {
(1,1.35063992884845)
(2,1.56706952937804)
(4,1.46703916415947)
(8,1.32319411138144)
(16,1.20352297014961)
(32,1.35157011732848)
(64,1.9459458281786)
};
\addplot[fill=SOFT] coordinates {
(1,1.40723611851711)
(2,1.6119254655509)
(4,1.5328512211937)
(8,1.41700064316729)
(16,1.30080061625654)
(32,1.41156166045214)
(64,1.94241416210442)
};
\end{axis}
\end{tikzpicture}
\end{subfigure}
\caption{List's Throughput with Range of 256}
\label{graph.list_256_threads}
\end{subfigure}
\begin{subfigure}[t]{\textwidth}
    \begin{subfigure}[t]{0.515\linewidth}
\begin{tikzpicture}
    \begin{axis}[
    xlabel={\#Threads},
    legend style={font=\tiny},
    ylabel={MOPS},
    xmin=0, xmax=70,
    ymin=0, ymax=14,
    xtick={0,10,20,30,40,50,60,70},
    ytick={0,2,4,6,8,10,12,14},
    legend pos= south east,
]
 
\addplot[
    color=LogFree,
    mark=x,
    error bars/.cd, y dir=both,y explicit
    ]
    coordinates {
(1,0.5163163)+-(0.00077824481355104,0.00077824481355104)
(2,1.012096)+-(0.00239379159345398,0.00239379159345398)
(4,1.907014)+-(0.00414535232216318,0.00414535232216318)
(8,3.58103142857143)+-(0.0152606136884973,0.0152606136884973)
(16,6.24988375)+-(0.0584377842612742,0.0584377842612742)
(32,8.226384)+-(0.0226519133341728,0.0226519133341728)
(64,8.330442)+-(0.293620382767267,0.293620382767267)
    };
    
\addplot[
    color=LinkFree,
    mark=o,
    error bars/.cd, y dir=both,y explicit
]
coordinates {
(1,0.7231344)+-(0.000451866412322383,0.000451866412322383)
(2,1.402882)+-(0.00619663679057274,0.00619663679057274)
(4,2.718728)+-(0.0167946665653244,0.0167946665653244)
(8,4.952332)+-(0.0256571467854786,0.0256571467854786)
(16,8.825182)+-(0.0193626492479742,0.0193626492479742)
(32,12.25741)+-(0.0228435905512144,0.0228435905512144)
(64,13.3007)+-(0.029627261369507,0.029627261369507)
};

\addplot[
    color=SOFT,
    mark=square,
    error bars/.cd, y dir=both,y explicit
]
coordinates {
(1,0.5328765)+-(0.00012669844147734,0.00012669844147734)
(2,1.040058)+-(0.00687603013370235,0.00687603013370235)
(4,2.016444)+-(0.0040583730882529,0.0040583730882529)
(8,3.754819)+-(0.00670825579227925,0.00670825579227925)
(16,7.080542)+-(0.00937290456632194,0.00937290456632194)
(32,10.8823)+-(0.0128067163321446,0.0128067163321446)
(64,12.34259)+-(0.0117247022416775,0.0117247022416775)
};
 
\end{axis}
\end{tikzpicture}
\end{subfigure}
    \begin{subfigure}[t]{0.45\linewidth}
\begin{tikzpicture}
\begin{axis}[
    ybar,
    ylabel={Relative Improvement},
    symbolic x coords={1,2,4,8,16,32,64},
    xlabel={\#Threads},
    xtick=data,
    bar width=\barwidth,
    nodes near coords align={vertical},
    xtick align=inside,
    ]
\addplot[color = LogFree,
        pattern=north east lines
    ] coordinates {(1,1)
(2,1)
(4,1)
(8,1)
(16,1)
(32,1)
(64,1)
};
\addplot[color=LinkFree] coordinates{
(1,1.40056473134782)
(2,1.38611554635133)
(4,1.42564658675815)
(8,1.38293452564744)
(16,1.41205538423015)
(32,1.49001189343945)
(64,1.5966379695099)
};
\addplot[fill = SOFT] coordinates {
(1,1.03207375014114)
(2,1.02762781396231)
(4,1.05738290332425)
(8,1.04853003244875)
(16,1.13290779208493)
(32,1.32285339463852)
(64,1.48162486456301)
};
\end{axis}
\end{tikzpicture}
\end{subfigure}
\caption{List's Throughput with Range of 1024}
\label{graph.list_1024_threads}
\end{subfigure}
\begin{subfigure}[t]{\textwidth}
    \begin{subfigure}[t]{0.515\linewidth}
\begin{tikzpicture}
    \begin{axis}[
    xlabel={\#Threads},
    ylabel={MOPS},
    xmin=0, xmax=70,
    ymin=0, ymax=210,
    xtick={0,10,20,30,40,50,60,70},
    ytick={0,30,60,90,120,150,180,210},
    legend style={at={(0.5,-0.3)},
      anchor=north,legend columns=-1},
]
 
\addplot[
    color=LogFree,
    mark=x,
    error bars/.cd, y dir=both,y explicit,
    ]
coordinates {
(1,1.738369)+-(0.0120723927230163,0.0120723927230163)
(2,3.278224)+-(0.0137645438092109,0.0137645438092109)
(4,6.197966)+-(0.0195012544239999,0.0195012544239999)
(8,12.0054833333333)+-(0.0999530437355428,0.0999530437355428)
(16,21.720325)+-(0.303900344283967,0.303900344283967)
(32,32.02774)+-(0.196382985472405,0.196382985472405)
(64,8.131385)+-(2.64298455367621,2.64298455367621)
    };
    \addlegendentry{Log-Free}
    
\addplot[
    color=LinkFree,
    mark=o,
    error bars/.cd, y dir=both,y explicit,
]
coordinates {
(1,4.150558)+-(0.00422193941673585,0.00422193941673585)
(2,7.897277)+-(0.0392266733988156,0.0392266733988156)
(4,15.62303)+-(0.0204914637042053,0.0204914637042053)
(8,29.79038)+-(0.0438285825215023,0.0438285825215023)
(16,55.79954)+-(0.0967565430428931,0.0967565430428931)
(32,104.5926)+-(0.123781576750238,0.123781576750238)
(64,189.5669)+-(0.221247510565083,0.221247510565083)
};
\addlegendentry{Link-Free}

\addplot[
    color=SOFT,
    mark=square,
    error bars/.cd, y dir=both,y explicit,
]
coordinates {
(1,4.014681)+-(0.0044130027794997,0.0044130027794997)
(2,7.700124)+-(0.0127053165039193,0.0127053165039193)
(4,15.35201)+-(0.0381133052565284,0.0381133052565284)
(8,29.77929)+-(0.0450719060020927,0.0450719060020927)
(16,57.10533)+-(0.0900632466405301,0.0900632466405301)
(32,108.6705)+-(0.171297266380518,0.171297266380518)
(64,199.3399)+-(0.586235552467297,0.586235552467297)
};
\addlegendentry{SOFT}
 
\end{axis}
\end{tikzpicture}
\end{subfigure}
    \begin{subfigure}[t]{0.45\linewidth}
\begin{tikzpicture}
\begin{axis}[
    ybar,
    legend style={at={(0.5,-0.3)},
      anchor=north,legend columns=-1},
    ylabel={Relative Improvement},
    xlabel={\#Threads},
    symbolic x coords={1,2,4,8,16,32,64},
    xtick=data,
    ymin=0,
    ymax=25,
    ytick={0,3,6,9,12,15,18,21,24},
    nodes near coords align={vertical},
    bar width=\barwidth,
    xtick align=inside,
    ]
\addplot[color = LogFree,
        pattern=north east lines
    ] coordinates{
(1,1)
(2,1)
(4,1)
(8,1)
(16,1)
(32,1)
(64,1)
};
\addplot[color = LinkFree] coordinates {
(1,2.3876162080663)
(2,2.40901079364924)
(4,2.52067049093202)
(8,2.48139780572488)
(16,2.56900115444866)
(32,3.26568780688241)
(64,23.3129903454332)
};
\addplot[fill = SOFT] coordinates {
(1,2.30945271113325)
(2,2.34887060798774)
(4,2.47694324234757)
(8,2.48047406115817)
(16,2.62911949982332)
(32,3.39301180788904)
(64,24.5148766169601)
};
\legend{Log-Free, Link-Free, SOFT}
\end{axis}
\end{tikzpicture}
\end{subfigure}
\caption{Hash Table's Throughput with Range of 1M}
\label{graph.hash_threads}
\end{subfigure}
\caption{Throughput as a Function of the \#Threads}
\label{graph.multi_threads}
\end{figure}

The results for the scalability test are displayed in Figure~\ref{graph.multi_threads}. On the left, the graphs show the throughput as a function of the number of threads (in millions of operations per second).
On the right, the improvement relative to log-free set is depicted (the y axis is the improvement factor).

In Figures~\ref{graph.list_256_threads} and \ref{graph.list_1024_threads}, we can see the results for the shorter and longer lists.
When the key range is 256 keys, all algorithms experience a peak with 16 threads and a slow decrease towards 64 threads.
For a single thread, \SOFT{} and link-free outperform log-free by 40\% and 35\%, respectively, for 16 threads by 30\% and 20\%, respectively, and for 64 threads, both by 94\%.
The 16-thread peak can be explained by the nature of a list.
Running many threads on a short list implies contention that hurts performance.
Also, 16 threads can use a single processor but 17 cannot.

\textsc{Soft} achieves the best performance on the short list by a noticeable margin.
In this case, the amount of \code{psync} operations dominates performance as the traversal times are short.
Unlike link-free or log-free, \SOFT{} uses the optimal number of fences per update.
For instance, both link-free and log-free executed a \code{psync} before trimming a logically deleted node (\SOFT{} does not).
Both of our algorithms perform much better than log-free and we can relate this result to the elimination of pointer flushing, which is the main idea behind both algorithms.

For a longer list (Figure~\ref{graph.list_1024_threads}), all the compared lists scale with the additional threads.
When the number of available keys is bigger, most of the time is spent on traversing the list; hence, more threads imply more concurrent traversals and more operations.

As can be seen in the graph, link-free outperforms both \SOFT{} and log-free by a considerable difference.
In contrast to Figure~\ref{graph.list_256_threads}, here the additional overhead of \SOFT{} (using intermediate states and more CAS-es instead of direct marking) degrades its performance.
When the range grows, the additional \code{psync} operations are masked by the traversal times.
Since \SOFT{} uses two additional CAS-es in each update, link-free wins.

Moreover, with higher contention, a node might be flushed more than once in link-free.
As mentioned, link-free prevents redundant \code{psync} operations using a flag after the first necessary \code{psync}.
In a case where multiple threads operate on the same key, it might be flushed more than needed.
So, when contention is high, link-free may perform more \code{psync} operations.
For cases of lower contention, the optimization is more effective.
In effect, link-free does a single \code{psync} per update and zero per read (due to the low contention, all flags are set before other threads help).
In this case, link-free and \SOFT{} execute the same amount of \code{psync} operations, but \SOFT{} is more complicated and uses more CAS-es.
Because of this, for boarder ranges, link-free performs better.

The hash set is evaluated in Figure~\ref{graph.hash_threads}.
Link-free and \SOFT{} are highly scalable (reaching 25.2x and 27x with 32 threads, respectively, and 45.6x and 49.6x with 64 threads, respectively).
Log-free is a lot less scalable (18.4x with 32 threads and 4.6x with 64 threads).
For 32 threads, \SOFT{} and link-free perform better by factors of 3.4x and 3.26x, respectively.
Thus, we obtain a dramatic improvement of the state-of-the-art.

As can be seen, the result of the log-free hash table in the test with 64 threads is oddly low.
We used the authors' implementation and we do not know why this happened.
To make further comparisons fair enough to previous work, we fixed the number of threads at 32 in subsequent hash table evaluations.
The number of threads in the lists' evaluation remained 64. 

\begin{figure}
\begin{subfigure}[t]{\textwidth}
\begin{subfigure}[t]{0.515\linewidth}
    \begin{tikzpicture}
\begin{semilogxaxis}[
    xlabel={Key Range},
    ylabel={MOPS},
    xmin=10, xmax=20000,
    ymin=0, ymax=20,
    xtick={10,100,1000,10000},
    ytick={0,5,10,15,20},
    legend pos=south east,
]
 
\addplot[
    color=LogFree,
    mark=x,
    error bars/.cd, y dir=both,y explicit,
    ]
    coordinates {
(16,6.313657)+-(0.534707471733553,0.534707471733553)
(64,8.46842285714286)+-(0.974069558673894,0.974069558673894)
(256,9.051633)+-(0.588323674757565,0.588323674757565)
(1024,8.35603857142857)+-(0.482625345448654,0.482625345448654)
(4096,2.85248666666667)+-(0.0205296615014676,0.0205296615014676)
(16384,0.754832375)+-(0.00433257050092618,0.00433257050092618)

    };
    
\addplot[
    color=LinkFree,
    mark=o,
    error bars/.cd, y dir=both,y explicit,
]
coordinates {
(16,14.44974)+-(0.0315095610624212,0.0315095610624212)
(64,17.78343)+-(0.0437245361269626,0.0437245361269626)
(256,17.83441)+-(0.0274244295639048,0.0274244295639048)
(1024,13.29214)+-(0.0544771103561947,0.0544771103561947)
(4096,5.900796)+-(0.0531470884332558,0.0531470884332558)
(16384,0.9053665)+-(0.034964199225378,0.034964199225378)
};

\addplot[
    color=SOFT,
    mark=square,
    error bars/.cd, y dir=both,y explicit,
]
coordinates {
(16,15.57707)+-(0.0327265969189625,0.0327265969189625)
(64,19.32487)+-(0.0301728771347499,0.0301728771347499)
(256,17.82062)+-(0.0199573773886201,0.0199573773886201)
(1024,12.33977)+-(0.012534243959686,0.012534243959686)
(4096,5.088652)+-(0.0104969965240812,0.0104969965240812)
(16384,0.7419869)+-(0.0222281438924007,0.0222281438924007)
};
 
\end{semilogxaxis}
\end{tikzpicture}
    \end{subfigure}
\begin{subfigure}[t]{0.45\linewidth}
\begin{tikzpicture}
\begin{axis}[
    ybar,
    legend style={at={(0.5,-0.3)},
      anchor=north,legend columns=-1},
    ylabel={Relative Improvement},
    xlabel={Key Range},
    symbolic x coords={16,64,256,1K,4K,16K},
    xtick=data,
    ymin=0.5,
    nodes near coords align={vertical},
    bar width=\barwidth,
    xtick align=inside,
    ]
\addplot[color = LogFree,
        pattern=north east lines
    ] coordinates{
(16,1)
(64,1)
(256,1)
(1K,1)
(4K,1)
(16K,1)
};
\addplot[color = LinkFree] coordinates {
(16,2.28864824300718)
(64,2.09996953387846)
(256,1.97029751427173)
(1K,1.59072267155985)
(4K,2.06864980963978)
(16K,1.19942722382569)
};
\addplot[fill = SOFT] coordinates {
(16,2.46720244701288)
(64,2.28199162063572)
(256,1.96877403226578)
(1K,1.47674880800489)
(4K,1.78393542008961)
(16K,0.982982347570876)
};
\end{axis}
\end{tikzpicture}
\end{subfigure}
\caption{List's Throughput}
\label{graph.list_sizes}
\end{subfigure}
\begin{subfigure}[t]{\textwidth}
\begin{subfigure}[t]{0.515\linewidth}
\begin{tikzpicture}
\begin{semilogxaxis}[
    xlabel={Key Range},
    ylabel={MOPS},
    xmin=1000, xmax=10000000,
    ymin=0, ymax=220,
    xtick={10,100,1000,10000,100000,1000000,10000000},
    ytick={0,50,100,150,200},
    legend style={at={(0.5,-0.35)},
      anchor=north,legend columns=-1},
]
 
\addplot[
    color=LogFree,
    mark=x,
    error bars/.cd, y dir=both,y explicit,
    ]
    coordinates {
(1024,55.16775)+-(0.6917132283361,0.6917132283361)
(16384,43.9542571428571)+-(0.401607857869339,0.401607857869339)
(262144,33.9635)+-(0.804715352698311,0.804715352698311)
(4194304,31.9869857142857)+-(0.31853818813772,0.31853818813772)
    };
    \addlegendentry{Log-Free}
    
\addplot[
    color=LinkFree,
    mark=o,
    error bars/.cd, y dir=both,y explicit,
]
coordinates {
(1024,182.1162)+-(0.45539946462183,0.45539946462183)
(16384,186.0508)+-(0.309373722999692,0.309373722999692)
(262144,132.9658)+-(0.142190198544193,0.142190198544193)
(4194304,99.86234)+-(0.111950230249564,0.111950230249564)
};
\addlegendentry{Link-Free}

\addplot[
    color=SOFT,
    mark=square,
    error bars/.cd, y dir=both,y explicit,
]
coordinates {
(1024,194.7893)+-(0.634622063928928,0.634622063928928)
(16384,206.1358)+-(0.435431826477301,0.435431826477301)
(262144,140.9938)+-(0.162554450709559,0.162554450709559)
(4194304,105.156)+-(0.0773934079536462,0.0773934079536462)
};
\addlegendentry{SOFT}
 
\end{semilogxaxis}
\end{tikzpicture}
    \end{subfigure}
\begin{subfigure}[t]{0.45\linewidth}
\begin{tikzpicture}
\begin{axis}[
    ybar,
    legend style={at={(0.5,-0.35)},
      anchor=north,legend columns=-1},
    ylabel={Relative Improvement},
    xlabel={Key Range},
    symbolic x coords={1K,16K,256K,4M},
    xtick=data,
    nodes near coords align={vertical},
    bar width=\barwidth,
    xtick align=inside,
    ]
\addplot[color = LogFree,
        pattern=north east lines
    ] coordinates{
(1K,1)
(16K,1)
(256K,1)
(4M,1)
};
\addplot[color = LinkFree] coordinates {
(1K,3.30113517408268)
(16K,4.23282776444863)
(256K,3.91496165000662)
(4M,3.12196781816176)
};
\addplot[fill = SOFT] coordinates {
(1K,3.53085453004699)
(16K,4.68978008956064)
(256K,4.15133304871406)
(4M,3.28746199905408)
};
\legend{Log-Free, Link-Free, SOFT}
\end{axis}
\end{tikzpicture}
\end{subfigure}
\caption{Hash Tables' Throughput}
\label{graph.hash_sizes}
\end{subfigure}
\caption{Throughput as a Function of Key Range}
\end{figure}

In the second experiment, we examined the effect of different key ranges on the performance of the data structure.
We again fixed the workload to be 90\% read operations, and the number of threads at 64 for the lists and at 32 for the hash maps.
The sizes when running the lists vary from 16 to 16K in multiples of 4.
For hash tables, the size varies between 1K and 4M in multiples of 16.

Figure~\ref{graph.list_sizes} shows that \SOFT{} and link-free are superior to log-free in each key range.
As expected, for shorter ranges, \SOFT{} performs better and for bigger ranges link-free wins.
The reason is that as the key range grows, more time is spent on traversals of the lists and the number of \code{psync} operations used is masked.
We can see this effect in the graph: as the range grows, the difference in performance shrinks, starting with a factor of 2.46x difference between \SOFT{} and log-free and ending with link-free having a 20\% improvement for 16K keys.

As expected, the trend of the graph consists of a single peak point.
We note that the performance improves because contention drops when the range grows but only up to a point.
Beyond this point, most of the time is spent on traversing the list rather than executing actual operations.

Figure~\ref{graph.hash_sizes} depicts the performance of the three hash tables and the improvement relative to log-free.
As explained above, this test was run with 32 threads.
As predicted, the performance of all hash tables worsens as the range grows.
This may be attributed to reduced locality.
For 1K distinct keys, \SOFT{} outperforms log-free by a factor of 3.53x and link-free outperforms log-free by a factor of 3.2x.
For the longest range (4M keys), \SOFT{} is better by a factor of 3.28x and link-free is better by a factor of 3.12x.

\begin{figure}
\begin{subfigure}[t]{\textwidth}
\begin{subfigure}[t]{0.515\linewidth}
\begin{tikzpicture}
    \begin{axis}[
    xlabel={\%Reads},
    ylabel={MOPS},
    ymin=0, ymax=90,
    xmin=50, xmax=100,
    ytick={0,10,20,30,40,50,60,70,80,90},
    xtick={50,60,70,80,90,100},
]
 
\addplot[
    color=LogFree,
    mark=x,
    error bars/.cd, y dir=both,y explicit,
    ]
    coordinates {
(50,3.044799)+-(0.273147629066349,0.273147629066349)
(60,3.86756222222222)+-(0.409029523358266,0.409029523358266)
(70,4.15925333333333)+-(0.369822465843176,0.369822465843176)
(80,6.00171888888889)+-(0.523011238686492,0.523011238686492)
(90,9.21719666666667)+-(0.469845401820294,0.469845401820294)
(95,15.05692)+-(0.838915126976308,0.838915126976308)
(100,63.58635)+-(0.743613024791129,0.743613024791129)
    };
    
\addplot[
    color=LinkFree,
    mark=o,
    error bars/.cd, y dir=both,y explicit,
]
coordinates {
(50,7.919239)+-(0.0225630009095185,0.0225630009095185)
(60,8.823912)+-(0.00783916919765603,0.00783916919765603)
(70,10.13429)+-(0.0204196390743059,0.0204196390743059)
(80,12.39005)+-(0.0211429353983364,0.0211429353983364)
(90,17.85001)+-(0.020977847150187,0.020977847150187)
(95,26.33649)+-(0.0576548761262932,0.0576548761262932)
(100,84.78409)+-(3.53616404947698,3.53616404947698)
};

\addplot[
    color=SOFT,
    mark=square,
    error bars/.cd, y dir=both,y explicit,
]
coordinates {
(50,7.832174)+-(0.0110425080204557,0.0110425080204557)
(60,8.700902)+-(0.0112332532475386,0.0112332532475386)
(70,10.026404)+-(0.0208564212906948,0.0208564212906948)
(80,12.29131)+-(0.0197071824361568,0.0197071824361568)
(90,17.81122)+-(0.0229832184636144,0.0229832184636144)
(95,26.34169)+-(0.0651534969121752,0.0651534969121752)
(100,74.29904)+-(1.33090339837349,1.33090339837349)
};
 
\end{axis}
\end{tikzpicture}
\end{subfigure}
\begin{subfigure}[t]{0.45\linewidth}
\begin{tikzpicture}
\begin{axis}[
    ybar,
    legend style={at={(0.5,-0.25)},
      anchor=north,legend columns=-1},
    ylabel={Relative Improvement},
    xlabel={\%Reads},
    symbolic x coords={50,60,70,80,90,95,100},
    xtick=data,
    nodes near coords align={vertical},
    bar width=\barwidth,
    xtick align=inside,
    ]
\addplot[color = LogFree, pattern=north east lines] coordinates{
(50,1)
(60,1)
(70,1)
(80,1)
(90,1)
(95,1)
(100,1)
};
\addplot[color = LinkFree] coordinates {
(50,2.60090698926267)
(60,2.28151778639775)
(70,2.43656473597825)
(80,2.06441691611681)
(90,1.93659858257699)
(95,1.74912863985463)
(100,1.33336934735207)
};
\addplot[fill = SOFT] coordinates {
(50,2.57231232669217)
(60,2.24971222182449)
(70,2.41062594568256)
(80,2.04796496263015)
(90,1.93239014465353)
(95,1.74947399600981)
(100,1.16847468049353)
};
\end{axis}
\end{tikzpicture}
\end{subfigure}
\caption{List's Throughput with Range of 256}
\label{graph.list_256_reads}
\end{subfigure}
\begin{subfigure}[t]{\textwidth}
\begin{subfigure}[t]{0.515\linewidth}
\begin{tikzpicture}
    \begin{axis}[
    legend style={font=\tiny},
    ylabel={MOPS},
    xlabel={\%Reads},
    ymin=0, ymax=25,
    xmin=50, xmax=100,
    ytick={0,5,10,15,20,25},
    xtick={50,60,70,80,90,100},
    legend pos= north west,
]
 
\addplot[
    color=LogFree,
    mark=x,
    error bars/.cd, y dir=both,y explicit,
    ]
    coordinates {
(50,3.275752)+-(0.32783682142265,0.32783682142265)
(60,3.86552285714286)+-(0.358018200071173,0.358018200071173)
(70,4.860223)+-(0.312167239598246,0.312167239598246)
(80,5.81937777777778)+-(0.354172981110511,0.354172981110511)
(90,8.25880428571428)+-(0.602857747605733,0.602857747605733)
(95,10.0984571428571)+-(0.264856426407674,0.264856426407674)
(100,18.2879375)+-(0.212200794419962,0.212200794419962)
    };
    
\addplot[
    color=LinkFree,
    mark=o,
    error bars/.cd, y dir=both,y explicit,
]
coordinates {
(50,7.194257)+-(0.0107316833296274,0.0107316833296274)
(60,7.874833)+-(0.0186473721089165,0.0186473721089165)
(70,8.798358)+-(0.0221010756377255,0.0221010756377255)
(80,10.29176)+-(0.0172538213980998,0.0172538213980998)
(90,13.29528)+-(0.0255452049754183,0.0255452049754183)
(95,16.75033)+-(0.0413565617630137,0.0413565617630137)
(100,22.62162)+-(0.343638332762657,0.343638332762657)
};

\addplot[
    color=SOFT,
    mark=square,
    error bars/.cd, y dir=both,y explicit,
]
coordinates {
(50,6.878772)+-(0.0121444167997434,0.0121444167997434)
(60,7.499011)+-(0.00968763619462425,0.00968763619462425)
(70,8.359847)+-(0.0189321887104923,0.0189321887104923)
(80,9.709498)+-(0.0148499229372324,0.0148499229372324)
(90,12.33213)+-(0.02051439218701,0.02051439218701)
(95,15.22045)+-(0.0215285491111236,0.0215285491111236)
(100,20.7074)+-(0.233762121048246,0.233762121048246)
};
\end{axis}
\end{tikzpicture}
\end{subfigure}
\begin{subfigure}[t]{0.45\linewidth}
\begin{tikzpicture}
\begin{axis}[
    ybar,
    legend style={at={(0.5,-0.25)},
      anchor=north,legend columns=-1},
    ylabel={Relative Improvement},
    xlabel={\%Reads},
    symbolic x coords={50,60,70,80,90,95,100},
    xtick=data,
    nodes near coords align={vertical},
    bar width=\barwidth,
    xtick align=inside,
    ]
\addplot[color = LogFree,
        pattern=north east lines
    ] coordinates{
(50,1)
(60,1)
(70,1)
(80,1)
(90,1)
(95,1)
(100,1)
};
\addplot[color = LinkFree] coordinates {
(50,2.19621540336387)
(60,2.03719737045367)
(70,1.81027866416829)
(80,1.76853271827458)
(90,1.60983110145831)
(95,1.65870189505611)
(100,1.23696945049161)
};
\addplot[fill = SOFT] coordinates {
(50,2.09990622000689)
(60,1.9399732654906)
(70,1.72005420327421)
(80,1.6684770040325)
(90,1.49321010322663)
(95,1.50720548542069)
(100,1.13229827037631)
};
\end{axis}
\end{tikzpicture}
\end{subfigure}
\caption{List's Throughput with Range of 1024}
\label{graph.list_1024_reads}
\end{subfigure}
\begin{subfigure}[t]{\textwidth}
\begin{subfigure}[t]{0.515\linewidth}
\begin{tikzpicture}
    \begin{axis}[
    xlabel={\%Reads},
    ylabel={MOPS},
    xmin=50, xmax=100,
    ymin=0, ymax=150,
    xtick={50,60,70,80,90,100},
    ytick={0,50,100,150},
    legend style={at={(0.5,-0.3)},
      anchor=north,legend columns=-1},
]
 
\addplot[
    color=LogFree,
    mark=x,
    error bars/.cd, y dir=both,y explicit,
    ]
    coordinates {
(50,12.3746857142857)+-(2.40377753917209,2.40377753917209)
(60,14.9318666666667)+-(1.80714340649954,1.80714340649954)
(70,17.6292666666667)+-(1.31648649976503,1.31648649976503)
(80,23.6492625)+-(0.49556538675889,0.49556538675889)
(90,32.1294444444444)+-(0.243866354640379,0.243866354640379)
(95,41.3353625)+-(0.475983931216771,0.475983931216771)
(100,74.0478625)+-(0.118682354161277,0.118682354161277)
    };
    \addlegendentry{Log-Free}
    
\addplot[
    color=LinkFree,
    mark=o,
    error bars/.cd, y dir=both,y explicit,
]
coordinates {
(50,65.7808)+-(0.0667118231436814,0.0667118231436814)
(60,72.63782)+-(0.111968799095928,0.111968799095928)
(70,81.88515)+-(0.14088513524038,0.14088513524038)
(80,93.34042)+-(0.215266149480071,0.215266149480071)
(90,108.5103)+-(0.196137469586582,0.196137469586582)
(95,118.2725)+-(0.323131376511449,0.323131376511449)
(100,130.0672)+-(0.180180058014557,0.180180058014557)
};
\addlegendentry{Link-Free}

\addplot[
    color=SOFT,
    mark=square,
    error bars/.cd, y dir=both,y explicit,
]
coordinates {
(50,65.4041)+-(0.125408978459133,0.125408978459133)
(60,73.25344)+-(0.217643983445876,0.217643983445876)
(70,83.49296)+-(0.153439449306407,0.153439449306407)
(80,96.90269)+-(0.143541789221272,0.143541789221272)
(90,114.2443)+-(0.177228503261594,0.177228503261594)
(95,125.5365)+-(0.162975823252854,0.162975823252854)
(100,139.5182)+-(0.337719890632564,0.337719890632564)
};
\addlegendentry{SOFT} 
\end{axis}
\end{tikzpicture}
\end{subfigure}
\begin{subfigure}[t]{0.45\linewidth}
\begin{tikzpicture}
\begin{axis}[
    ybar,
    legend style={at={(0.5,-0.3)},
      anchor=north,legend columns=-1},
    ylabel={Relative Improvement},
    xlabel={\%Reads},
    symbolic x coords={50,60,70,80,90,95,100},
    xtick=data,
    nodes near coords align={vertical},
    bar width=\barwidth,
    xtick align=inside,
    ymin=0,
    ]
\addplot [color = LogFree,
        pattern=north east lines
    ] coordinates{
(50,1)
(60,1)
(70,1)
(80,1)
(90,1)
(95,1)
(100,1)
};
\addplot[color = LinkFree] coordinates {
(50,5.31575520532701)
(60,4.86461750707659)
(70,4.64484153245172)
(80,3.94686388211895)
(90,3.3772852869469)
(95,2.86129098299307)
(100,1.75652875867956)
};
\addplot[fill = SOFT] coordinates {
(50,5.28531402817734)
(60,4.90584610988579)
(70,4.73604271684585)
(80,4.0974931036433)
(90,3.55575086888109)
(95,3.03702429124699)
(100,1.88416242264927)
};
\legend{Log-Free, Link-Free, SOFT}
\end{axis}
\end{tikzpicture}
\end{subfigure}
\caption{Hash Table's Throughput with Range of 1M}
\label{graph.hash_reads}
\end{subfigure}
\caption{Throughput as a Function of the Percentage of Reads}
\end{figure}

The last variable evaluated is the workload.
We measured different distributions of reads (50\% -- 100\% with increments of 10\%, and also the specific values of 95\%).
Note that this covers the standard ``Yahoo! Cloud Serving Benchmark'' (YCSB) \citep{cooper2010benchmarking} workloads A (50\% reads), B (95\% reads), and C (100\% reads). 
In this experiment, the number of threads was fixed at 64 for lists and 32 threads for hash tables, and the key ranges were fixed at 256 or 1024 in the case of the lists and at 1M in the case of the hash tables.

The lists (Figures~\ref{graph.list_256_reads} and \ref{graph.list_1024_reads}) all behaved similarly to one another.
For both ranges, link-free performed slightly better than \SOFT.
Link-free is superior to \SOFT{} since the high amount of threads increases the contention, which increases the cost of the additional CAS-es used in \SOFT{}.
Also, a higher percentage of updates also contributed to more CAS-es in \SOFT{}.

For the shorter range, link-free surpassed log-free by a factor of 2.6x with 50\% reads, and for 100\% reads, it had a 33\% improvement.
With 1k keys, the throughput of link-free was higher by a factor of 2.1x with 50\% reads and higher by 23\% with 100\% reads.

The trend of both graphs can be justified by a few reasons.
First, all algorithms use the least amount of \code{psync} operations in the read operations.
\textsc{Soft} does not use any, link-free uses at most one, and log-free uses at most two.
Moreover, reads are faster since there is no need to invalidate any cache lines of other processors.
Finally, unlike insert and remove, which may restart and theoretically run forever, the contains operation is wait-free and optimized to run as fast as possible.
Accordingly, the gap between the different algorithms shrinks as the percentage of reads grows.

Running with 100\% reads is a special situation where the performance improves tremendously.
Each thread runs in isolation from the others since there are no conflicts between contains operations.
Also, in this case, none of the algorithms execute any \code{psync} operations.
Link-free and log-free both use optimizations to reduce the number of \code{psync} operations and since the nodes in the list were inserted and flushed previous to the beginning of the test, there is no need to flush them again.

We would expect \SOFT{} to be the best in this scenario but due to its implementation, it falls short.
Unlike link-free, each volatile node in \SOFT{} has an additional pointer that makes it larger.
As a result, about one and a half volatile nodes fit in a single cache line, so when traversing the list, we have more cache misses.
\textsc{Soft} is still better than log-free because its contains operation is simpler.
Log-free has a few branches to check whether a node should be flushed or not, which lengthens the function and may cause branch mis-predictions.

The hash tables, depicted in Figure~\ref{graph.hash_reads}, exhibit a trend similar to what we saw in previous tests.
The throughput rises as the number of updates declines.
Moreover, the difference in performance between the three algorithms shrinks as the number of updates decreases.

In according with our expectations, \SOFT{} surpasses link-free and log-free.
The traversal times in the hash tables are minimal so \SOFT{} does not suffer from cache misses and the simplistic contains operation works in \SOFT's favor.
\section{Related Work}\label{sec.related_work}

There has been a lot of research focused on adapting specific concurrent data structures to durable ones \citep{schwalb2015nvc, Friedman:3178487.3178490, David:215967, nawab2017dali}.
Some researchers developed techniques to modify general objects into durable linearizable ones \citep{coburn2012nv, Cohen:3210377.3210400, 10.1007/978-3-662-53426-7_23, volos2011mnemosyne, kolli2016high, avni2016persistent}

\citet{coburn2012nv, volos2011mnemosyne, kolli2016high} used transactions to create a new interface to the NVRAM and, by proxy, make regular objects durable linearizable.
The main disadvantage of their schemes is the need to log operations and other kinds of metadata in the NVRAM, which causes more explicit writes to the memory and uses of synchronization primitives.
Another major disadvantage is the use of locks that limits the scalability of the different implementations and might cause an unbounded rollback effect upon a crash.

\citet{10.1007/978-3-662-53426-7_23} presented a general algorithm to maintain durable linearizability.
This generality, however, comes at the expense of efficiency; their construction inserts a fence before every shared write and a flush after, a fence and a \code{psync} for each CAS, and a \code{psync} after every shared read.
In contrast, our algorithms are optimized in the sense they execute fewer \code{psync} operations, especially \SOFT{}.

\citet{Cohen:3152284.3133891} presented a sequential durable hash table that uses only one \code{psync} per update and none for reads, achieving the lower bound proven by \citet{Cohen:3210377.3210400}.
This paper introduced the validity schemes we used in both algorithms.
Both algorithms rely on the observation made in the paper that the order of writes to the \emph{same cache line} in the program is the same as the order of those writes in the memory.
No extension to concurrency was discussed in their paper.

\citet{nawab2017dali} developed an efficient hash table that supports multiple threads and transactions.
They used fine-grained synchronization, and thus their algorithm is not lock-free.
Their algorithm does not support durable linearizability but only buffered durable linearizability which is a weaker guarantee.
Thus this work is not comparable to ours.

\citet{Friedman:3178487.3178490} presented three variations of a durable lock-free queue.
The first guarantees durable linearizability \citep{10.1007/978-3-662-53426-7_23},
the second guarantees \emph{detectable execution} \citep{Friedman:3178487.3178490}, which is a stronger guarantee than durable linearizability, and the third guarantees buffered durable linearizability \citep{10.1007/978-3-662-53426-7_23}.
The queue is inherently different from a set since it maintains an order between individual keys.

\citet{Cohen:3210377.3210400} introduced a theoretical universal construct to obtain durable lock-free objects with one \code{psync} per update (per conflicting thread) and none for reads.
Their implementation uses a lock-free queue to order all pending operations, then a batch of operations is persisted together and, finally, a flag is set to indicate that the operations were flushed.
This algorithm is theoretical and is not targeted at high performance.
Using a queue to order operations creates contention and hurts scalability.
In addition, the state of the object is a persistent log of all the previous operations, which means that in order to return a result, the whole log has to be traversed, making this algorithm highly inefficient and impractical.

\citet{David:215967} introduced four kinds of sets (\emph{Log-Free Data Structures}), building up from lock-free data structures and adding to them two main optimizations.
\emph{Link-and-persist} is the first optimization and it reduces the number of \code{psync} operations but at the cost of using CAS, which is considered more expensive than a simple \code{store} operation \citep{david2013everything}.
The second is \emph{link-cache}, which writes \emph{next} pointers to the NVRAM only when another operation depends on the persistency of the pointer.
This work represents state-of-the-art durable sets and we compared our constructions to it, showing dramatic improvements.

\section{Conclusion}\label{sec.conclusion}
In this work we presented two algorithms for durable lock-free sets: link-free and \SOFT. 
These two algorithms were shown to outperform existing state-of-the-art by significant factors of up to \maxImprovement{}.
In addition to high efficiency, they also demonstrated excellent scalability.  
The main idea underlying these algorithms was to avoid persisting the data structure's pointers, at the expense of reconstructing the data structures during (infrequent) recoveries from crashes. 
\SOFT{} reduces fences to the minimum theoretical value, at the expense of algorithmic complication and higher (volatile) synchronization. 
The evaluation demonstrated that \SOFT{} outperforms the link-free implementation when \code{psync} operations are often required: For example, for long lists it was better to use the link-free version because traversals were long and \code{psync} operations were infrequent. 
For short lists (which also underlay a hash table), however, operations are short and \code{psync} operations occur frequently.
In this case, \SOFT{} was the best performing method.  

\bibliography{main.bib}

\appendix
\section{Proof Terminology}\label{chap:terminology}
In this section we present the terminology we use in order to prove that the implementations of the link free and soft lists are correct.
The proof for the link-free list appears in Section~\ref{chap:link_free_correctness_gali} and the proof for the soft list appears in Section~\ref{chap:soft_correctness_gali}.
We assume a sequentially consistent execution in both proofs, i.e., we do not consider reordering of loads and stores in our proof.
This simplifies the proofs and is also unavoidable because there is no formal definition of C++11 extended to include \code{psync} operations.

We use the notion of \emph{durable linearizability}~\citep{10.1007/978-3-662-53426-7_23} for correctness.
To show that the proposed lists are durable linearizable, we define linearization points for all of the completed operations (either by returning or after a recovery procedure), as well as some of the operations which are still pending during a crash event.
Note that in our setting, following a crash, recovery is always completed before any new operation begins executing, so there exists a point in time where recovery is completed and no new operation began executing.

We consider a standard shared memory setting with a bounded collection of threads,
communicating via shared memory space~\citep{lynch1996distributed}. 
Data structures are accessed via their operations. 
A data structures includes a (finite or infinite) set of possible configuration (with a distinguished initial configuration) and a sequential specification.  
The \textit{sequential specification} defines the expected behavior of the data structure in a sequential execution.
It specifies, per operation, the preconditions and their respective postconditions. 
It also takes into account the operation's input parameters.
We consider deterministic data structures for which, given an operation and a current configuration (that obeys the operation's preconditions), there exists a single resulting configuration that satisfies the operation's postconditions.

A {\em configuration} at any point during the execution \textit{state} is an instantaneous snapshot of the system, specifying the content of the shared memory space as well as local variables and state for each thread. 
In the \textit{initial configuration} the shared data structure and the local variable are all in their initial state, and each thread's program counter is the set to the beginning of its program.
The threads change the system's configuration by taking \textit{steps}, which may include performing some local computations and, additionally, an operation on the shared memory.
In particular, an invocation of an operation (together with the respective input parameters) and the return from an operation are each considered to be a single step.
We assume each step is atomic.
An \textit{execution} consists of an alternating sequence of configurations and steps, starting with the initial state.
These steps are called {\em events}. 
We often consider only the execution steps, since the configurations can be easily computed given the execution steps. 
Each step is coupled with the executing thread (its identity depends on the \textit{scheduler}).
If it is an invocation of an operation, it also includes the input parameters, and if it is the last step of the operation, it is associated with the operation result.
This last step is denoted the {\em return} or the {\em response} of the operation.
A sub-execution of a given execution $E$ is a sub-sequence of events in $E$.

\subsection{Linearizability}\label{sec:linearizability} 
For each invocation of an operation in an execution we can match a response that follows it in the execution with the same invoking thread and operation type.
We can also match execution steps of this operations (by the same thread) in between the two. 
A invocation is called \emph{pending} in execution $E$ if its matching response does not exist.
We denote \emph{complete($E$)} as a sub-execution of $E$ containing all pairs of matching invocations and responses,
and all execution steps executing these operations.
For each thread $T$ we define a sub-execution by using the notation $E|T$.
This sub-execution is the sub-sequence of steps of operations performed by thread $T$.
Two executions, $E, \hat{E}$ are equivalent if for every thread $T$ they are the same, $E|T = \hat{E}|T$.
A execution $S$ is called \emph{sequential} if it begins with an invocation and a new invocation in the 
execution occurs immediately after the response to the previous invocation occurred. 
We call a execution \emph{well-formed} if for every thread $T$, $E|T$ is sequential. 

We say that an operation $op_0$ \emph{happens before} an operation $op_1$ in a execution $E$, if the response of $op_0$ occurs in $E$ before $op_1$'s invocation. In this case, we write $op_0 \prec_E op_1$.
The relation \emph{happens before} defines a partial order of a execution $E$. 

\emph{Sequential specification}. Given a data structure, we define a \emph{sequential specification} as a set of sequential executions that are legal for this data structure. 

\begin{definition}[Linearizability~\citep{Herlihy:1990:LCC:78969.78972}]
	\label{def:linearizability}
	A execution $E$ of data structure operations is \emph{linearizable} if it can be extended by adding zero or more responses at the end of $E$ to create $\hat{E}$, and there exists a legal sequential execution $S$ for the data structure such that:
	\begin{enumerate}
		\item $complete(\hat{E})$ is equivalent to $S$
		\item if $m_0 \prec_{\hat{E}} m_1$, then  $m_0 \prec_{S} m_1$ in $S$.
	\end{enumerate}
	We refer to $S$ as the linearization of $E$, and the order of the operations in $S$ is called the linearization order.
\end{definition}

Informally, given a execution of invocations, some of which might be pending, and responses, a execution is \emph{linearizable} if there exists a {\em linearization order $L$} of some of the operations in $E$ that satisfies the following.
All complete operations appear in $L$, where complete means that the operation's invocation has a matching response
in the execution.
Some pending invocations are completed by the responses added in $\hat{E}$ and these must also appear in $L$. 
The remaining pending invocations do not appear in $L$.
If we execute the operations according to their linearization order sequentially, we get the same execution 
(equivalent parameters and results), and the sequential execution is legal for the data structure. 

Often, papers use an equivalent definition, in which the second condition is replaced as follows. 
We associate a {\em linearization point} for each operation with a step in the execution, and require that the linearization point of each operation appears between its invocation and its response.
The linearization order is then determined by the order of the linearization points. 
It is easy to see that adequate linearization points satisfy the second condition. 
Showing that adequate linearization points can be specified for each adequate linearization is
also not hard. 

In this paper we sometimes view a linearization of an execution as a specification of linearization
points and at other times as an order of the operations that satisfies Condition (2).
These views are equivalent. 

\subsection{Durable Linearizability}\label{sec-durable-linearizability}
Moving to durable linearizability, we add crashes as new types of events in a execution and a new correctness criterion is required.
We use the notion of \emph{durable linearizability}~\citep{10.1007/978-3-662-53426-7_23} for correctness.
In the setting of durable linearizability, after a crash, new threads are created to continue the tasks of the program.
A recovery procedure is executed after each crash by the new threads to complete some of the pending operations 
that were executed before the crash.
The recovery is completed before new operations start executing. 
Loosely speaking, we require that the execution, minus all crash events and minus operations that did not survive the crashes, be linearizable. 
Given an execution $E$ of data structure operations, we denote $ops(E)$ the sub-execution where all the crashes are omitted. 

\begin{definition}[Durable Linearizability] \label{def:dl}
A execution $E$ is said to be \emph{durably linearizable} if it is well-formed (i.e., for every thread $T$, $E|T$ is sequential) and ops(E) is linearizable.
A data structure is durable linearizable if any concurrent execution of data structure operations yields a durably linearizable execution. 
\end{definition}

Note that after a crash, new threads are generated to continue the program execution.
Therefore, any operation that executes concurrently with a crash is a pending operation that can never complete.
It cannot have a matching response in the execution because the thread that invoked it does not execute beyond the crash.
In contrast, any operation that completes before the crash has a matching response that completes it.
Definition~\ref{def:dl} requires a linearization of all operations in the execution.
By Definition~\ref{def:linearizability} (linearizability), this means that all completed operations must be included in the linearization order, but operations that executed concurrently with a crash are pending at the end of the execution 
and they do not need to be in the linearization order.
Some of these operations may be matched with a response at the end of the execution and be included in the linearization order, whereas other operations may remain pending and be left out of the linearization order. 

In our proof, we specify a linearization order for executions of the proposed durable lists. 
We specify which of the operations that execute concurrently with the crash survive, 
i.e., are matched with a response at the end of the execution and are included in the linearization order.
We denote these operations as operations that {\em survive} the crash. 

The recovery procedure, executed after a crash (and described in Section~\ref{sub_sec:lf_recovery} and~\ref{sub_sec:soft_recovery}), is assumed to terminate before new threads start executing their code.
Given an operation for which a crash event occurs after its invocation and before its response, we consider its response point as the end of the respective recovery procedure.
Notice that in the following definitions, we do not consider recoveries that are interrupted by crash events. We do so for clarity and brevity. The definitions can be easily extended to include such cases.
\section{Link-Free Correctness}
\label{chap:link_free_correctness_gali}
We start by proving some basic list invariants.
In Section~\ref{chap:link_free_dl} we prove the linearizability of our implementation when there are no crash events, and that it is also durable linearizable.
Finally, we show that our implementation is lock-free in Section~\ref{chap:link_free_lock_freedom}.

The content of a node in the volatile memory, can be different from its content in the NVRAM, due to modifications that have not been persisted yet (either by implicit or explicit flushes).
We distinguish between the two representations of a single link-free node: the \emph{volatile node}, and the \emph{persistent copy} which contains only the modifications written back to the NVRAM (implicitly or explicitly).

We start by stating some basic definitions we are going to use throughout our proof.
Notice that, unless stated otherwise, the definitions relate to the volatile nodes (regardless of being written to the non-volatile memory).

\begin{definition}[Reachability] \label{definition:link_free_reachability}
	We say that a node $n$ is \emph{reachable} from a node $n'$ if there exists nodes $n_0,n_1,\ldots,n_k$ such that $n_0=n'$, $n_k=n$ and for every $0 \leq i <k$, $n_i$ is the predecessor of $n_{i+1}$ (via its \emph{next} pointer).
	We say that a node $n$ is \emph{reachable} if it is reachable from the \code{head} sentinel node.
\end{definition}

\begin{definition}[Infant Nodes] \label{definition:link_free_infant}
	We say that a node $n$ is an \emph{infant} if $n$ is neither \code{head} nor \code{tail}, and there does not exist an earlier successful execution of the CAS operation in line~\ref{code.insert_cas} of Listing~\ref{algo:lf_insert}, satisfying $newNode=n$.
\end{definition}

\begin{definition}[A Node's State] \label{definition:link_free_node_state}
Let $n$ be a node (which is neither \code{head} nor \code{tail}), and let $b$ be the initial value of its two validity bits.
\begin{enumerate}
    \item We say that $n$ is at its {\em initial state} if the value of both of its validity bits is $b$.
    \item We say that $n$ is {\em invalid} if the value of its first validity bit is $\neg b$ and the value of its second validity bit is $b$.
    \item We say that $n$ is {\em valid} if the value of both of its validity bits is $\neg b$.
    \item We say the $n$ is {\em marked} if its {\em next} pointer is marked. Otherwise, we say that $n$ is {\em unmarked}.
\end{enumerate}
The \code{head} and \code{tail} sentinel nodes are always considered as valid and unmarked nodes.
\end{definition}

We now prove some basic claims regarding the link-free list implementation.

\begin{claim}[State Transitions] \label{claim:link_free_state_changes}
Let $n$ be a volatile node. Then its state can only go through the following transitions:
	\begin{enumerate}
	    \item From being unmarked and in its initial stage, to being unmarked and invalid. \label{state:from_intial_to_invalid}
	    \item From being unmarked and invalid, to being unmarked and valid. \label{state:from_invalid_to_valid}
	    \item From being unmarked and valid, to being marked and valid. \label{state:from_unmarked_to_marked}
	\end{enumerate}
\end{claim}

\begin{proof} A node $n$ is always created with an initialized and unmarked state, and its state can only change in line~\ref{code.lf_node_exists} of Listing~\ref{algo:lf_contains}, line~\ref{code.lf_insert_node_exists}, \ref{code.make_node_invalid} or~\ref{code.insert_make_valid} of Listing~\ref{algo:lf_insert}, or in line~\ref{code.remove_make_valid} or~\ref{code.logically_remove} of listing~\ref{algo:lf_remove}.
As explained is Section~\ref{sub_sec.lf_help}, executing the \code{flipV1} or \code{makeValid} auxiliary functions on the same node more than once, would not effect its state. Moreover, when \code{makeValid} is executed before \code{flipV1}, the node's state remains initialized and is also not effected.
Therefore, \code{flipV1} only changes the node's state from being initialized to being invalid, \code{makeValid} only changes the node's state from being invalid to being valid, and it remains to show that the marking of a node does not foil the above transition types.
Since a node can only be marked (line~\ref{code.logically_remove} of listing~\ref{algo:lf_remove}), and is never unmarked throughout the execution, we only need to show that it is valid when marked.
If it is either invalid or valid before executing line~\ref{code.remove_make_valid}, then from the above, it is valid when marked in line~\ref{code.logically_remove}.
Notice that it also cannot be at its initialized state, since a node becomes invalid right after its creation, in line~\ref{code.make_node_invalid} of Listing~\ref{algo:lf_insert}.
\end{proof}

\begin{claim} [Marked Nodes] \label{claim:link_free_marked_node}
Once a node is marked, its \emph{next} pointer does not change anymore.
\end{claim}

\begin{proof} Let $n$ be a marked node. From Claim~\ref{claim:link_free_state_changes}, it cannot be unmarked. Besides when marked in line~\ref{code.logically_remove} of listing~\ref{algo:lf_remove}, $n$'s \emph{next} pointer can only change during a successful CAS execution in line~\ref{code.trim_unlink_cas} of Listing~\ref{algo:lf_help} or line~\ref{code.insert_cas} of Listing~\ref{algo:lf_insert}.
In both cases, it is assumed that $n$ is unmarked and therefore, the CAS execution is unsuccessful if it is marked, leaving $n$'s \emph{next} pointer unchanged.
\end{proof}

\begin{claim} [The States of the Sentinel Nodes] \label{claim:link_free_sentinel_state}
The \code{head} and \code{tail} sentinel nodes are always unmarked and valid.
\end{claim}

\begin{proof} As mentioned in the proof of Claim~\ref{claim:link_free_state_changes}, a node's state can only change when its key is sent as an input parameter to one of the list's operations. Assuming the neither $-\infty$ nor $\infty$ are sent as input parameters to the list's operations, the states of the \code{head} and \code{tail} sentinel nodes always remain unmarked and valid.
\end{proof}

\begin{claim} [Nodes Invariants] \label{claim:link_free_nodes}
Let $n_1$ and $n_2$ be two different nodes. Then:
\begin{enumerate}
    \item If $n_2$ is the successor of $n_1$ in the list then $n_2$ is not an infant. \label{link_free_node_inv:no_infant}
    
    \item Right before executing line~\ref{code.insert_cas} in Listing~\ref{algo:lf_insert}, having $newNode=n_2$, it holds that: (1) $n_2$ is an infant, and (2) $n_2$ is invalid. \label{link_free_node_inv:is_infant}
    
    \item If $n_2$ is not an infant and not marked, or marked but not yet flushed since being marked, then $n_2$ is reachable. \label{link_free_node_inv:is_reachable}
    
    \item If $n_2$ is marked, but has not been flushed since being marked, then $n_2$ is reachable. \label{link_free_node_inv:not_flushed_is_reachable}
    
    \item If $n_1$'s key is smaller than or equal to $n_2$'s key, then $n_1$ is not reachable from $n_2$. \label{link_free_node_inv:not_reachable}
    
    \item If $n_2$ is reachable from $n_1$ at a certain point, then as long as $n_2$ is not marked, $n_2$ is still reachable from $n_1$. \label{link_free_node_inv:still_reachable}
    
    \item If $n_1$ is not an infant then the \code{tail} sentinel node is reachable from $n_1$. \label{link_free_node_inv:tail_reachable}
\end{enumerate}
\end{claim}

\begin{proof}
We are going to prove the claim by induction on the length of the execution.
At the initial stage, \code{head} and \code{tail} are the only nodes in the list, having $-\infty$ and $\infty$ keys (respectively), both are reachable by Definition~\ref{definition:link_free_reachability}, and \code{head} is \code{tail}'s predecessor.
Therefore, all of the invariants obviously hold.
Now, assume that all of the invariants hold at a certain point during the execution, at let $s$ be the next execution step, executed by a thread $t$.
\begin{enumerate}
    \item If $n_2$ is not an infant before executing $s$, then by Definition~\ref{definition:link_free_infant}, it is not an infant after executing $s$, and the invariant holds. Otherwise, by the induction hypothesis, $n_2$ does not have a predecessor before executing $s$, and it cannot be the \code{head} sentinel node.
    $n_1$'s successor can only change in line~\ref{code.trim_unlink_cas} of Listing~\ref{algo:lf_help}, or in line~\ref{code.init_next} or~\ref{code.insert_cas} of Listing~\ref{algo:lf_insert}. If $s$ is the execution of line~\ref{code.trim_unlink_cas} in Listing~\ref{algo:lf_help} or line~\ref{code.init_next} in Listing~\ref{algo:lf_insert}, then $n_2$ has already been traversed during a former \code{find} execution, as a node with a predecessor, and by the induction hypothesis, is not an infant. If $s$ is the execution of line~\ref{code.insert_cas} in Listing~\ref{algo:lf_insert}, then $n_2$ is not an infant by Definition~\ref{definition:link_free_infant}.
    \item Since $n_2$ can only be that node during the execution of the \code{insert} operation in which it is created, and which returns in line~\ref{code.insert_last_line}, after a successful CAS execution in line~\ref{code.insert_cas}, by Definition~\ref{definition:link_free_infant}, $n_2$ must be an infant at this point, and (1) holds. 
    Now, assume by contradiction that $n_2$ is not invalid. Since it becomes invalid in line~\ref{code.make_node_invalid} and by Claim~\ref{claim:link_free_state_changes}, its state must be valid. $n_2$'s state can become valid only in line~\ref{code.lf_node_exists} of Listing~\ref{algo:lf_contains}, in line~\ref{code.lf_insert_node_exists} or~\ref{code.insert_make_valid} of Listing~\ref{algo:lf_insert}, or in line~\ref{code.remove_make_valid} of Listing~\ref{algo:lf_remove}. In all cases, it must have a predecessor prior to that change, and by invariant~\ref{link_free_node_inv:no_infant}, it is not an infant -- a contradiction. Therefore, $n_2$'s state is invalid, and the invariant holds.
    \item If $n_2$ was an infant before executing $s$, then $s$ is the execution of line~\ref{code.insert_cas} in Listing~\ref{algo:lf_insert}, making $n_2$ the successor of some node which is reachable by assumption. $n_2$ is reachable in this case. Otherwise, by assumption and Claim~\ref{claim:link_free_state_changes}, it was reachable right before executing $s$. Assume by contradiction that it is no longer reachable after executing $s$. 
    Then $n_2$ is reachable from a node $n_1$ that was reachable right before $s$, and is no longer reachable (may be $n_2$ itself). Assume w.l.o.g that $n_1$ is such a node for which the path of nodes from Definition~\ref{definition:link_free_reachability} is the longest. The node $n_1$ can only become unreachable if the current step is the execution of line~\ref{code.trim_unlink_cas} in Listing~\ref{algo:lf_help}, and if $n_1$ is marked and then flushed in line~\ref{code.trim_flush} of Listing~\ref{algo:lf_help}. This means that $n_1 \neq n_2$. Since $n_1$'s successor stays reachable in this case, we get a contradiction. Therefore, $n_2$ is reachable in this case as well.
    \item By assumption, $n_1$ is not reachable from $n_2$ right before executing $s$. Since all changes of nodes' successors (line~\ref{code.trim_unlink_cas} of Listing~\ref{algo:lf_help}, and line~\ref{code.init_next} and~\ref{code.insert_cas} of Listing~\ref{algo:lf_insert}) preserve keys order (notice the halting condition in line~\ref{code.lf_if_bigger} of Listing~~\ref{algo:lf_help}), the Invariant still holds.
    \item If $n_2$ is not reachable from $n_1$ before executing $s$ then the invariant holds vacuously. Otherwise, assume by contradiction that $n_2$ was reachable from $n_1$ right before executing $s$, and is no longer reachable from $n_1$ after executing it. Let $n_3$ be the first node reachable from $n_1$ after the previous step, that is not reachable from it after executing the current step ($n_3$ must exist). 
    The node $n_3$ can only become unreachable from $n_1$ if the current step is the execution of line~\ref{code.trim_unlink_cas} in Listing~\ref{algo:lf_help}, and if $n_3$ is marked. This means that $n_3 \neq n_2$. Since $n_3$'s successor stays reachable from $n_1$ in this case, we get a contradiction. Therefore, $n_2$ is still reachable from $n_1$.
    \item If $n_1$ was an infant right before executing $s$ then $s$ is executing a successful CAS in line~\ref{code.insert_cas} of Listing~\ref{algo:lf_insert}. In this case, $s$ makes $n_1$ the predecessor of a node whose \code{tail} is reachable from, by assumption. Therefore, \code{tail} is reachable from $n_1$ in this case. Otherwise, assume by contradiction that \code{tail} was reachable from $n_1$ right before executing $s$ (must hold by assumption), but is no longer reachable from it after executing it. Let $n_2$ be the last node reachable from $n_1$, for whom \code{tail} is not reachable from after executing the current step ($n_2$ must exist).
    Then the current step must change $n_2$'s \emph{next} pointer. Since $n_2$ cannot be an infant (by Invariant~\ref{link_free_node_inv:no_infant}), this step is a successful CAS, either in line~\ref{code.trim_unlink_cas} of Listing~\ref{algo:lf_help} or in line~\ref{code.insert_cas} of Listing~\ref{algo:lf_insert}.
    In both cases, $n_2$'s successor is set to be a node that \code{tail} is reachable from, by assumption. Since we get a contradiction to Definition~\ref{definition:link_free_reachability}, \code{tail} is reachable from $n_1$ in this case as well.
\end{enumerate}
\end{proof}

\begin{claim} [The Volatile List Invariant] \label{claim:link_free_volatile_sorted} The list is always sorted by the nodes' keys, no key ever appears twice, and the \code{head} and \code{tail} sentinel nodes are always the first and last members of the list, respectively. 
\end{claim}

\begin{proof} From Invariant~\ref{link_free_node_inv:not_reachable} of Claim~\ref{claim:link_free_nodes}, the volatile list is always sorted by the nodes' keys and no key ever appears twice. By Claim~\ref{claim:link_free_sentinel_state} and Invariant~\ref{link_free_node_inv:is_reachable} of Claim~\ref{claim:link_free_nodes}, the \code{head} and \code{tail} sentinel nodes are always members of the list, and by Invariant~\ref{link_free_node_inv:not_reachable} of Claim~\ref{claim:link_free_nodes}, they are the first and last members, respectively.
\end{proof}

We now move to dealing with the persistent list. The persistent list contains the persistent copies of the volatile list's nodes, as long as their state is valid and not marked, as stated in Definition~\ref{definition:link_free_persistently} below.

\begin{definition} [Persistently in the List] \label{definition:link_free_persistently} Let $n$ be a node. We say that $n$ is {\em persistently in the list} if the state of $n$'s persistent copy is valid and not marked.

\end{definition}

Claim~\ref{claim:link_free_persistently} below asserts that being valid and not marked is sufficient for staying persistently in the list. In particular, the \code{head} and \code{tail} sentinel nodes always remain persistently in the list.

\begin{claim} [Being Persistently in the List] \label{claim:link_free_persistently} Let $n$ be a node which is persistently in the list. As long as $n$'s state is valid and unmarked, $n$ is still persistently in the list.
\end{claim}

\begin{proof} Assume that $n$ is persistently in the list at some point, and let assume by contradiction that there exists a later point, in which $n$'s state is valid and unmarked, and is not persistently in the list. We are going to consider the earliest such point. 
By Claim~\ref{claim:link_free_state_changes}, $n$ does not change between the mentioned two points. Therefore, each flush of $n$, flushes it with a valid and unmarked -- a contradiction. Thus, $n$ is still persistently in the list.
\end{proof}

\begin{claim} [Persistently in the List Nodes are Reachable] \label{claim:link_free_if_persistently_then_reachable} Let $n$ be a node which is persistently in the list. Then $n$ is reachable.
\end{claim}

\begin{proof}
Assume that $n$ is persistently in the list. By Definition~\ref{definition:link_free_persistently}, during the last flush of $n$ to the non-volatile memory, $n$'s state was valid and unmarked.
If $n$ is unmarked, then by Invariants~\ref{link_free_node_inv:is_infant} and~\ref{link_free_node_inv:is_reachable} of Claim~\ref{claim:link_free_nodes}, $n$ is reachable.
Otherwise, since $n$ is marked but still persistently in the list, $n$ has not been flushed in line~\ref{code.trim_flush} of Listing~\ref{algo:lf_help}, line~\ref{code.contains_flush_marked} of Listing~\ref{algo:lf_contains}, or implicitly flushed yet, and in particular, it has not become unreachable in line~\ref{code.trim_unlink_cas} of Listing~\ref{algo:lf_help} yet (according to the proof of Claim~\ref{claim:link_free_nodes}, it cannot become unreachable in other scenarios).
Therefore, $n$ is still reachable in this case as well.
\end{proof}

Notice that Claim~\ref{claim:link_free_if_persistently_then_reachable} does not hold temporarily during recovery, until the list is reconstructed. However, this fact does not effect the use of this claim throughout our proof.

\begin{claim} [The Persistent List is a Set] \label{claim:link_free_set} The persistent list never contains two different persistent nodes with the same key.
\end{claim}

\begin{proof} The claim derives directly from Claim~\ref{claim:link_free_volatile_sorted} and~\ref{claim:link_free_if_persistently_then_reachable}.
\end{proof}

\begin{claim} \label{claim:link_free_find_linearization} Let $n_1$ and $n_2$ be the two volatile nodes returned as output from the find method. Then during the method execution, there exist a point in which (1) $n_1$ is reachable, (2) $n_2$ is $n_1$'s successor, and (3) $n_2$ is unmarked.
\end{claim}

\begin{proof}
When $n_1$'s marked bit is read for the first time during the execution, it is unmarked (otherwise, it would have been trimmed and not returned). In addition, since it must have had a predecessor at an earlier point (otherwise, it would not have been traversed), from Invariant~\ref{link_free_node_inv:no_infant} of Claim~\ref{claim:link_free_nodes}, it is not an infant, and from Invariant~\ref{link_free_node_inv:is_reachable} of Claim~\ref{claim:link_free_nodes}, it is reachable at this point.
If $n_2$ is $n_1$ successor at this point, then the claim holds for this point. Notice that $n_2$ cannot be marked at this point, since otherwise, it would have been trimmed at a later point and not returned as output.
If $n_2$ is not $n_1$'s successor at this point, then there exists a point between the first read of $n_1$ and the first read of $n_2$ in which $n_2$ becomes $n_1$'s successor. From Claim~\ref{claim:link_free_marked_node}, $n_1$ is unmarked at this point and thus, from Invariant~\ref{link_free_node_inv:is_reachable} of Claim~\ref{claim:link_free_nodes}, it is reachable at this point. In addition, $n_2$ is unmarked at this point as well, and the claim holds in this case.
\end{proof}

\begin{claim} \label{claim:link_free_unsucc_insert}
Let there be an insert execution that returns false in line~\ref{code.lf_insert_node_exists_returning_false} (Listing~\ref{algo:lf_insert}), and let $m$ be the node returned as the second output parameter from the last find call in line \ref{code.insert_find}.
Then at least one of the following holds during the insert execution:
\begin{enumerate}
    \item $m$ is persistently in the list.
    \item $m$ is marked and then flushed.
\end{enumerate}
\end{claim}

\begin{proof}
Claim~\ref{claim:link_free_find_linearization} guarantees that there exists a point during the last find execution in which $m$ is reachable and unmarked.
Since $m$ is made valid no later than the execution of line~\ref{code.lf_insert_node_exists}, and is flushed, while being still valid (by Claim~\ref{claim:link_free_state_changes}), no later than the execution of line~\ref{code.lf_insert_line_after_makevalid}, it is either persistently in the list (by Definition~\ref{definition:link_free_persistently}), or becomes marked before its flush. In both cases, the claim holds.
\end{proof}

\begin{claim} \label{claim:link_free_contains_point} Let $n_2$ be a node which is assigned into the \code{curr} variable in line~\ref{code.lf_end_loop} of Listing~\ref{algo:lf_contains}, and let $n_1$ be the last node assigned into the \code{curr} variable before $n_2$. Then there exists a point during the traversal in which both nodes are reachable and $n_2$ is $n_1$'s successor.
\end{claim}

\begin{proof}
Assume by contradiction that the claim does not hold. W.l.o.g., Let $n_1$ and $n_2$ be the first two nodes for which (1) $n_1$ and $n_2$ are assigned into the \code{curr} variable sequentially, and (2) the guaranteed point does not exist for them. Since this point does exist for $n_1$ and the former node assigned into \code{curr}, $n_1$ is reachable at some point during the execution (if $n_1$ is the \code{head} sentinel node then it is obviously reachable). From Invariant~\ref{link_free_node_inv:still_reachable} of Claim~\ref{claim:link_free_nodes}, $n_1$ is reachable as long as it is not marked. Since $n_2$ is its successor when assigned into the \code{curr} variable, from Claim~\ref{claim:link_free_marked_node} it was its successor at the last step in which $n_1$ was reachable before this assignment (might be the assignment itself). Therefore, there exists such a point for $n_1$ and $n_2$ -- a contradiction, and the claim holds. 
\end{proof}
\subsection{Durable Linearizability}\label{chap:link_free_dl}
We use the notion of \emph{durable linearizability}~\citep{10.1007/978-3-662-53426-7_23} for correctness.
The recovery procedure, executed after a crash (and described in Section~\ref{sub_sec:lf_recovery}), is assumed to terminate before new threads start executing their code.
Given an operation for which a crash event occurs after its invocation and before its response, we consider its response point as the end of the respective recovery procedure.
Notice that in the following definitions, we do not consider recoveries that are interrupted by crash events.
We do so for clarity and brevity.
The definitions can be easily extended to include such cases.

We are going to prove that, given an execution, removing all crash events would leave us with a linearizable history, including all the operations that were fully executed between two crashes, and some of the operations that were halted due to crashes (and then recovered during recovery).
We are going to define, per operation execution, whether it is {\em a surviving operation}.
A surviving operation is an operation that is linearized in the final crash-free history of the execution (by removing all crash events).
Obviously, operations that were fully executed between two crash events are always considered as surviving operations.
Additionally, we are going to define the linearization points of all surviving operations in the crash-free history.

For each linearized operation, we define its linearization point as a point during its execution in which it takes effect.
For a more accurate definition, we first define, in Definition~\ref{definition:link_free_dl_set_member} below, which nodes are considered as set members.
Given this definition, a successful insertion takes effect when a respective new node becomes a set member, a successful removal takes effect when an existing respective set member is removed from the set, a contains execution returns an answer which respects the set membership definition, and unsuccessful operations fail according to this definition as well.

\begin{definition} [Being a Set Member] \label{definition:link_free_dl_set_member}
Given a node $n$, it is considered as a set member as long as at least one of the following holds:
\begin{enumerate}
    \item $n$ is persistently in the list according to Definition~\ref{definition:link_free_persistently}.
    \item If $n$ is marked and then flushed, for the first time since it becomes valid, then $n$ is considered as a set member during the period in which it is valid and not yet flushed.
\end{enumerate}
\end{definition}

Notice that being persistently in the list is not effected by crash events (since it depends on the state saved in the non-volatile memory).
Moreover, a node which is considered as a set member of the second type, stops being a set member before the next crash event.
Therefore, being a set member is well-defined, even in the presence of crash events.
For using the term of set membership in our durable linearizability proof, we still need to prove that the collection of all set members is indeed a set.
We do so in Claim~\ref{claim:link_free_indeed_a_set}.

\begin{claim}\label{claim:link_free_indeed_a_set}
Let $n_1$ and $n_2$ be two different set members.
Then $n_1$'s key is different from $n_2$'s key.
\end{claim}

\begin{proof}
Assume by contradiction that $n_1$ and $n_2$ are two different set members with the same key.
By Claim~\ref{claim:link_free_set}, the persistent list never contains two different persistent nodes with the same key, and therefore, at least one of them is not persistently in the list. Assume, w.l.o.g., that $n_1$ is not persistently in the list.

Since $n_1$ is a set member, by Definition~\ref{definition:link_free_dl_set_member}, it is valid, and either not marked, or marked and not flushed yet. By Invariant~\ref{link_free_node_inv:is_reachable} of Claim~\ref{claim:link_free_nodes}, $n_1$ is reachable.
By Claim~\ref{claim:link_free_volatile_sorted}, there cannot be two reachable nodes with the same key, and therefore, $n_2$ is not reachable, and by Invariant~\ref{link_free_node_inv:is_reachable} of Claim~\ref{claim:link_free_nodes}, it is either not valid, or marked and flushed. In both cases, it is not a set member according to Definition~\ref{definition:link_free_dl_set_member} -- a contradiction. Therefore, there cannot exist two different set members with the same key, and the claim follows.
\end{proof}

We are now going to define, per operation, the terms for being considered as a surviving operation (in the presence of a crash event), its respective linearization point. In addition, we are going to prove that each survivng operation indeed takes effect at its linearization point, and that non-surviving operations do not take effect at all. 

\subsubsection{Insert}
Before defining the conditions for the survival of an insert operation, we need to re-define the success of an insertion in the presence of crash events.

\begin{definition} [A Successful Insert Operation] \label{definition:link_free_dl_successful_insert}
Given an execution of an insert operation, we say that this operation is successful if one of the following holds before any crash event, following its invocation:
\begin{enumerate}
    \item The operation returns \code{true}.
    \item A node $n$ is allocated in line~\ref{code.link_free_alloc_node}, becomes valid, and is flushed afterwards (not necessarily in the scope of the operation in which it is allocated).
\end{enumerate}
The operation is unsuccessful if it returns \code{false}.
\end{definition}

\begin{definition} [A Surviving Insert Operation] \label{definition:link_free_dl_surviving_insert}
An insert operation is considered as a surviving operation if, before the first crash event that follows its invocation, one of the following holds:

\begin{enumerate}
    \item The operation is unsuccessful according to Definition~\ref{definition:link_free_dl_successful_insert}. Let $m$ be the node returned as the second output parameter from the last find call in line \ref{code.insert_find}. The operation's linearization point is set to be a point, during the execution, in which $m$ is a set member according to Definition~\ref{definition:link_free_dl_set_member} (chosen arbitrarily).
    \item The operation is successful according to Definition~\ref{definition:link_free_dl_successful_insert}, and the node allocated in line~\ref{code.link_free_alloc_node} becomes persistently in the list (see Definition~\ref{definition:link_free_persistently}) before the crash event.
    In this case, the linearization point is set to be the flush which inserts it to the persistent list.
    \item The operation is successful according to Definition~\ref{definition:link_free_dl_successful_insert}, and the node allocated in line~\ref{code.link_free_alloc_node} does not become persistently in the list before the first crash event.
    In this case, the linearization point is set to be the step which changes its state to valid.
\end{enumerate}
\end{definition}

\begin{claim} \label{claim:link_free_dl_insert_linearization}
A surviving insert operation takes effect instantaneously at its linearization point. 
\end{claim}

\begin{proof} We are going to prove the claim for each of the three surviving insertion types.
\begin{enumerate}
    \item Suppose that the operation is unsuccessful according to Definition~\ref{definition:link_free_dl_successful_insert}, and let $m$ be the node returned as the second output parameter from the last find call in line \ref{code.insert_find}. Notice that $m$'s key is equal to the key received as input.
    We are going to show that $m$ is a set member at the linearization point defined in Definition~\ref{definition:link_free_dl_surviving_insert}, and therefore, the (unsuccessful) insert operation indeed takes effect at this point. 
    According to Claim~\ref{claim:link_free_unsucc_insert}, there must exist a point during the execution in which $m$ is either persistently in the list, or that it is marked and then flushed.
    In the first scenario, by Definition~\ref{definition:link_free_dl_set_member}, $m$ is indeed a set member, and we are done.
    In the second scenario, it is guaranteed by Claim~\ref{claim:link_free_find_linearization} that $m$ is marked during the execution (since it is valid and unmarked at some point during the find method execution). Therefore, a point at which it is a set member, exists according to Definition~\ref{definition:link_free_dl_set_member}, and the claim holds.
    \item Suppose that the operation is successful according to Definition~\ref{definition:link_free_dl_successful_insert}, and the node allocated in line~\ref{code.link_free_alloc_node} becomes persistently in the list (see Definition~\ref{definition:link_free_persistently}) before the crash event. By Definition~\ref{definition:link_free_dl_set_member}, the allocated node indeed becomes a set member at the linearization point defined above and thus, the operation takes effect instantaneously at this point.
    \item Suppose that the operation is successful according to Definition~\ref{definition:link_free_dl_successful_insert}, and the node allocated in line~\ref{code.link_free_alloc_node} does not become persistently in the list before the first crash event. By Definition~\ref{definition:link_free_dl_successful_insert}, it becomes valid, then marked, and then flushed, during the execution, and before any crash event. By Definition~\ref{definition:link_free_dl_set_member}, it becomes a set member when its state becomes valid and thus, the operation indeed takes effect at its defined linearization point.
\end{enumerate}
\end{proof}

\begin{claim} \label{claim:link_free_dl_nonsurviving_insert}
A non-surviving insert operation takes no effect. 
\end{claim}

\begin{proof} 
By Definition~\ref{definition:link_free_dl_surviving_insert}, during a none-surviving insert operation, if a volatile node is allocated, and even if it is inserted into the volatile list, and becomes valid, it is not flushed. By Definition~\ref{definition:link_free_dl_set_member}, it is not considered as a set member.
In particular, it is not persistently in the list and thus, will also not be considered as a set member after a crash event.
\end{proof}

\subsubsection{Remove} 
We also re-define the success of a removal in the presence of crash events.

\begin{definition} [A Successful Remove Operation] \label{definition:link_free_dl_successful_remove}
Given an execution of a remove operation, we say that this operation is successful if one of the following holds before any crash event, following its invocation:
\begin{enumerate}
    \item The operation returns \code{true}.
    \item A node $n$ is marked in line~\ref{code.logically_remove} and is flushed afterwards (not necessarily in the scope of the operation in which it is marked).
\end{enumerate}
The operation is unsuccessful if it returns \code{false}.
\end{definition}

\begin{definition} [A Surviving Remove Operation] \label{definition:link_free_dl_surviving_remove}
A remove operation is considered as a surviving operation if, before the first crash event that follows its invocation, one of the following holds:

\begin{enumerate}
    \item The operation is unsuccessful according to Definition~\ref{definition:link_free_dl_successful_remove}. The operation's linearization point is set to be the point guaranteed by Claim~\ref{claim:link_free_find_linearization}.
    \item The operation is successful according to Definition~\ref{definition:link_free_dl_successful_remove}. The operation's linearization point is set to be the first flush of the victim node, after its marking in line~\ref{code.logically_remove}.
\end{enumerate}
\end{definition}

\begin{claim} \label{claim:link_free_dl_remove_linearization}
A surviving remove operation takes effect instantaneously at its linearization point. 
\end{claim}

\begin{proof} We are going to prove the claim for each of the two surviving removal types.
\begin{enumerate}
    \item Suppose that the operation is unsuccessful according to Definition~\ref{definition:link_free_dl_successful_remove}, and let $m$ be the node returned as the second output parameter from the last find call in line \ref{code.remove_find}. Notice that $m$'s key is different from the key received as input.
    By Claim~\ref{claim:link_free_find_linearization}, $m$ is reachable at the linearization point. Moreover, its key is bigger than the key received as input, its predecessor's key is smaller than this key (by the find specification) and by Claim~\ref{claim:link_free_volatile_sorted}, there does not exist a reachable node with the input key. Since being a set member implies being reachable, there does not exist a set member with the given key at its linearization point and thus, it indeed takes effect this point.
    \item Suppose that the operation is successful according to Definition~\ref{definition:link_free_dl_successful_remove}, and the node marked in line~\ref{code.logically_remove} is flushed afterwards, and before the following crash event.
    If the non-volatile memory already contains a valid and unmarked copy of this node, then the operation's linearization point (according to Definition~\ref{definition:link_free_dl_surviving_remove}) indeed removes this node from the set, according to Definition~\ref{definition:link_free_dl_set_member}. Otherwise, the mentioned flush is the first flush of the victim node, and according to Definition~\ref{definition:link_free_dl_set_member}, it removes it from the set in this case as well.
\end{enumerate}
\end{proof}

\begin{claim} \label{claim:link_free_dl_nonsurviving_remove}
A non-surviving remove operation takes no effect. 
\end{claim}

\begin{proof} 
By Definition~\ref{definition:link_free_dl_surviving_remove}, during a none-surviving remove operation, if a victim node is found, and even if it is made valid and marked, it is not flushed. By Definition~\ref{definition:link_free_dl_set_member}, whether it is originally a set member or not, it is not removed from the set. 
\end{proof}

\subsubsection{Contains} We do not use the term of success for describing a contains execution, and, therefore, the terms for its survival are straight forward.

\begin{definition} [A Surviving Contains Operation] \label{definition:link_free_dl_surviving_contains}
A contains operation is considered as a surviving operation if it terminates before the first crash event that follows its invocation.
For defining linearization points per contains execution, let $n_1$ and $n_2$ be the last nodes assigned into the \code{curr} variable.
\begin{enumerate}
    \item When the operation returns \code{true}, its linearization point is set to be a point during the execution in which $n_2$ is a set member (chosen arbitrarily).
    \item When the operation returns \code{false} in line~\ref{code.lf_key_not_found_return_false}, its linearization point is set to be the point guaranteed by Claim~\ref{claim:link_free_contains_point} for $n_1$ and $n_2$.
    \item When the operation returns \code{false} in line~\ref{code.contains_flush_marked_return_false}, its linearization point is set to be a point during the execution in which $n_2$ is reachable but not a set member (chosen arbitrarily).
\end{enumerate}

\end{definition}

\begin{claim} \label{claim:link_free_dl_contains_linearization}
A surviving contains operation takes effect instantaneously at its linearization point. 
\end{claim}

\begin{proof} Let $n_1$ and $n_2$ be the last nodes assigned into the \code{curr} variable. We are going to prove the claim for each of the three surviving contains types.
\begin{enumerate}
    \item Suppose that the operation returns \code{true}. We are going to show that there indeed exists a point during the execution in which $n_2$ is a set member.
    Since the operation does not return in line~\ref{code.contains_flush_marked_return_false}, from Claim~\ref{claim:link_free_state_changes}, $n_2$ is not marked during the traversal. 
    In addition, since $n_2$ is made valid at the latest when executing line~\ref{code.lf_node_exists}, from Claim~\ref{claim:link_free_state_changes}, it is also valid when executing line~\ref{code.lf_contains_line_after_makevalid}. There are several possible scenarios:
    \begin{enumerate}
        \item $n_2$ is not marked during the execution. In this case, $n_2$ becomes a set member at the latest when executing line~\ref{code.lf_contains_line_after_makevalid}. In this case, there obviously exists a point during the execution at which $n_2$ is a set member.
        \item $n_2$ is marked during the execution, and is flushed at some point after becoming valid and before becoming marked (by Claim~\ref{claim:link_free_state_changes}, $n_2$ becomes valid before it is marked). There exists a suitable point in this case as well.
        \item The remaining case is when $n_2$ is not flushed after becoming valid and before being marked. In this case, it is flushed at the latest in line~\ref{code.lf_contains_line_after_makevalid}, and therefore, by Definition~\ref{definition:link_free_dl_set_member}, it is a set member at some point, before being marked.
    \end{enumerate}
    There exists a suitable linearization point in every case.
    \item Suppose that the operation returns \code{false} in line~\ref{code.lf_key_not_found_return_false}. Claim~\ref{claim:link_free_contains_point} guarantees that both $n_1$ and $n_2$ are reachable at this point. Since $n_1$'s key must be smaller then the key received as input, and $n_2$ must be bigger, by Claim~\ref{claim:link_free_volatile_sorted}, there does not exist a reachable node with the given key at this point. By Claim~\ref{claim:link_free_if_persistently_then_reachable}, there does not exist a set member with the given key at this point.
    \item Suppose that the operation returns \code{false} in line~\ref{code.contains_flush_marked_return_false}. If it still reachable when executing line~\ref{code.contains_flush_marked}, then a marked copy of $n_2$ resides in the non-volatile memory (i.e., it is not a set member by Definition~\ref{definition:link_free_dl_set_member}), while $n_2$ is still reachable, and the guaranteed point exists.
    Otherwise, before it becomes unreachable (which happens during the contains execution, according to Claim~\ref{claim:link_free_contains_point}), at the latest, it is flushed as a marked node in line~\ref{code.trim_flush} of Listing~\ref{algo:lf_help}. Therefore, the guaranteed point exists in this case as well.
    By Claim~\ref{claim:link_free_volatile_sorted} and~\ref{claim:link_free_if_persistently_then_reachable}, there does not exist a set member with the given key at this point.
\end{enumerate}
\end{proof}

Since a contains operation does not effect the list (it executes flushes, that can also be executed implicitly), there is no need to prove that non-surviving contains executions do not take effect.

\begin{theorem}\label{theorem:link_free_durable_linearizable}
The link-free list is durable linearizable. 
\end{theorem}

\begin{proof}
By Definition~\ref{definition:link_free_dl_surviving_insert}, \ref{definition:link_free_dl_surviving_remove} and~\ref{definition:link_free_dl_surviving_contains}, all the operations that are fully executed between two crashes (and some of the operations that are halted due to crash events), have a linearization point. 
By Claim~\ref{claim:link_free_dl_insert_linearization}, \ref{claim:link_free_dl_remove_linearization} and~\ref{claim:link_free_dl_contains_linearization}, each operation takes effect instantaneously at its linearization point. By Claim~\ref{claim:link_free_dl_nonsurviving_insert} and~\ref{claim:link_free_dl_nonsurviving_remove}, operations for which we did not define linearization points (non-surviving operations), do not take effect at all.
In summary, the link-free list is durable linearizable by definition~\citep{10.1007/978-3-662-53426-7_23}.
\end{proof}
\subsection{Lock-Freedom}\label{chap:link_free_lock_freedom}

\subsubsection{A Preliminary Discussion}
Lock-freedom is impossible to show in the presence of crashes. To see that this is the case, 
imagine an adversarial schedule of crashes that repeatedly creates a crash one step before the completion 
of an operation. Such crashes can also occur during the recovery process itself. As far as we know, lock-freedom has not
been previously discussed in the presence of crashes. 

One way to deal with this problem is to admit that in the presence of crashes lock-freedom cannot be guaranteed, 
but as crashes are expected to occur infrequently, this still leaves the question of lock-freedom during crash-free
executions. Such lock-freedom is of high value in practice, when crashes are indeed rare. 
A more theoretical approach is to consider 
crashes as progress, as if a crash itself is one of the operations on the data structure. Interestingly, this 
yields the same challenge. While executions with crashes always make progress, crash-free executions 
need a proof of progress. So in what follows we prove that the link-free list is lock-free in the 
absence of crashes. 

We are going to prove that in crash-free executions, at least one of the operations terminates.
To derive a contradiction, assume there is some execution for which no executing operation terminates after a certain point.
Notice that we can assume that no operation is invoked after this point, and that the set of running operations is finite (since there is a finite number of system threads).
The rest of the proof relates to the suffix $\alpha$ of the execution, starting from this point.

\begin{claim} \label{claim:link_free_no_state_changes}
There is a finite number of state changes of reachable nodes during $\alpha$.
\end{claim}

\begin{proof}
A contains execution must terminate after executing line~\ref{code.lf_node_exists}, an insert execution must terminate after executing line~\ref{code.lf_insert_node_exists} or~\ref{code.insert_make_valid}, and a remove execution must terminate after a successful CAS execution in line~\ref{code.logically_remove}.
In addition, the state change in line~\ref{code.make_node_invalid}, during an insert execution, is of an unreachable node.
Consequently, we can assume that after a certain point, state changes are made only in line~\ref{code.remove_make_valid}, of Listing~\ref{algo:lf_remove}. Since a finite number of new nodes is created and made reachable during $\alpha$ (at most one node per pending insert operation), and since every such node eventually becomes valid in line~\ref{code.insert_make_valid} of Listing~\ref{algo:lf_insert}, we can assume that the number of state changes in line~\ref{code.remove_make_valid} of Listing~\ref{algo:lf_remove} is finite as well. 
\end{proof}

\begin{claim} \label{claim:link_free_no_new_reachables}
There is a finite number of pointer changes of reachable nodes during $\alpha$.
\end{claim}

\begin{proof}
The pointers of reachable nodes change either in line~\ref{code.trim_unlink_cas} of Listing~\ref{algo:lf_help} or line~\ref{code.insert_cas} of Listing~\ref{algo:lf_insert}. A state change in line~\ref{code.insert_cas} of Listing~\ref{algo:lf_insert} would cause the termination of an insert execution and thus, the only pointer changes are physical removals of marked nodes, executed in line~\ref{code.trim_unlink_cas} of Listing~\ref{algo:lf_help}.
Since there is a finite number of state changes of reachable nodes during $\alpha$ (by Claim~\ref{claim:link_free_no_state_changes}), the number of marked nodes is bounded and thus, there is a finite number of pointer changes of reachable nodes during $\alpha$. 
\end{proof}

\begin{theorem} \label{theorem:link_free_lock_free}
The link-free list is lock-free. 
\end{theorem}

\begin{proof}
From Claims~\ref{claim:link_free_no_state_changes} and~\ref{claim:link_free_no_new_reachables}, after a certain point, there are no state or pointer changes in the list.
Therefore, we consider the suffix $\alpha'$ of the execution that contains no state or pointer changes of reachable nodes.
Obviously, starting from this point, the list becomes stable, and does not change anymore.

Since the list is finite, from Claim~\ref{claim:link_free_volatile_sorted}, every find and contains execution eventually ends.
In addition, every insert and remove operation must be unsuccessful, and also terminate (since calls to the find method always terminate).
We get a contradiction and therefore, the implementation is lock-free.
\end{proof}
\section{SOFT Correctness}
\label{chap:soft_correctness_gali}
In this section we prove the correctness (i.e., durable linearizability) and progress guarantee (lock-freedom) of the \SOFT{} list. We start by proving some volatile list invariants. In Section~\ref{chap:soft_linearizability} we prove the linearizability of our implementation when there are no crash events, followed by a durable linearizability proof in Section~\ref{chap:soft_dl}. Finally, we show our implementation is lock-free in Section~\ref{chap:soft_lock_freedom}.

\begin{claim}[State Transitions] \label{claim:soft_state_changes}
	The state of a volatile node can only go through the following transitions:
	\begin{enumerate}
	    \item From ``intend to insert'' to ``inserted'' \label{state:from_intend_to_insert}
	    \item From ``inserted'' to ``inserted with intention to delete'' \label{state:from_insert_to_intend}
	    \item From ``inserted with intention to delete'' to ``deleted'' \label{state:from_intend_to_delete}
	\end{enumerate}
\end{claim}

\begin{proof}
A node's state can change either in line~\ref{code.soft_complete_insert} of Listing~\ref{algo:soft_insert}, or in line~\ref{code.soft_mark_node} or \ref{code.soft_complete_remove} of Listing~\ref{algo:soft_remove}. In all three cases, the state changes according to one of the options mentioned above, and the claim follows immediately.
Notice that in the rest of the assignments into a node's \emph{next} pointer (line~\ref{code.soft_trim_cas} of Listing~\ref{algo:soft_find} and line~\ref{code.soft_link_node} of Listing~\ref{algo:soft_insert}), the state stays unchanged.
\end{proof}

\begin{claim} [Deleted States] \label{claim:soft_deleted_state}
Once the state of a node becomes ``deleted'', its \emph{next} pointer does not change anymore.
\end{claim}

\begin{proof}
A node's \emph{next} pointer changes either in line~\ref{code.soft_trim_cas} of Listing~\ref{algo:soft_find} or in line~\ref{code.soft_link_node} of Listing~\ref{algo:soft_insert}. In both cases, the state of the node whose \emph{next} pointer is to be updated, is checked before the update (guaranteeing that its state is not ``deleted''), and the CAS execution ensures that it does not change until the pointer changes (from Claim~\ref{claim:soft_state_changes}, its state cannot become ``deleted'' and change again afterwards).
Notice that we deal with state changes in Claim~\ref{claim:soft_state_changes}. In this claim we refer only to reference changes.
\end{proof}

\begin{claim} [The States of the Sentinel Nodes] \label{claim:soft_sentinel_state}
The states of the \code{head} and \code{tail} sentinel nodes are always ``inserted''.
\end{claim}

\begin{proof}
As mentioned in the proof of Claim~\ref{claim:soft_state_changes}, a node's state can change either in line~\ref{code.soft_complete_insert} of Listing~\ref{algo:soft_insert}, or in line~\ref{code.soft_mark_node} or \ref{code.soft_complete_remove} of Listing~\ref{algo:soft_remove}. In all three cases, the node's key is sent as an input parameter to the insert or remove operation, respectively. Assuming the neither $-\infty$ nor $\infty$ are sent as input parameters to the insert and remove operations, the states of the \code{head} and \code{tail} sentinel nodes always remain ``inserted''.
\end{proof}

\begin{definition}[Reachability] \label{definition:soft_reachability}
	We say that a volatile node $n$ is \emph{reachable} from a volatile node $n'$ if there exists nodes $n_0,n_1,\ldots,n_k$ such that $n_0=n'$, $n_k=n$ and for every $0 \leq i <k$, $n_i$ is the predecessor of $n_{i+1}$ in the list.
	We say that a node $n$ is \emph{reachable} if it is reachable from the \code{head} sentinel node.
\end{definition}

\begin{definition}[Logically in the List] \label{definition:soft_logically}
	We say that a volatile node $n$ is logically in the list if $n$ is reachable and its state is either ``inserted'' or ``inserted with intention to delete''.
\end{definition}

\begin{definition}[Infant Nodes] \label{definition:soft_infant}
	We say that a volatile node $n$ is an \emph{infant} if $n$ is neither \code{head} nor \code{tail}, and there does not exist an earlier successful execution of the CAS operation in line~\ref{code.soft_link_node} in Listing~\ref{algo:soft_insert}, satisfying $newNode=n$.
\end{definition}

\begin{claim} [Volatile Nodes Invariants] \label{claim:soft_nodes}
Let $n_1$ and $n_2$ be two different volatile nodes. Then:
\begin{enumerate}
    \item If $n_2$ is the successor of $n_1$ then $n_2$ is not an infant. \label{node_inv:no_infant}
    
    \item Right before executing line~\ref{code.soft_link_node} in Listing~\ref{algo:soft_insert}, having $newNode=n_2$, it holds that: (1) $n_2$ is an infant, and (2) $n_2$'s state is ``intend to insert''. \label{node_inv:is_infant}
    
    \item If $n_2$ is not an infant and its state is not ``deleted'', then $n_2$ is reachable. \label{node_inv:is_reachable}
    
    \item If $n_1$'s key is smaller than or equal to $n_2$'s key, then $n_1$ is not reachable from $n_2$. \label{node_inv:not_reachable}
    
    \item If $n_2$ is reachable from $n_1$ at a certain point, then as long as $n_2$'s state is not ``deleted'', $n_2$ is still reachable from $n_1$. \label{node_inv:still_reachable}
    
    \item If $n_1$ is not an infant then the \code{tail} sentinel node is reachable from $n_1$. \label{node_inv:tail_reachable}
\end{enumerate}
\end{claim}
 
\begin{proof}
In the initial stage, the \code{head} and \code{tail} sentinels are the only volatile nodes in the list, both with an ``inserted'' state, and \code{tail} is \code{head}'s successor. Invariant~\ref{node_inv:no_infant} holds since \code{tail} is not an infant, Invariant~\ref{node_inv:is_infant} holds vacuously, Invariants~\ref{node_inv:is_reachable}, \ref{node_inv:still_reachable} and~\ref{node_inv:tail_reachable} hold since both \code{head} and \code{tail} are reachable, and Invariant~\ref{node_inv:not_reachable} holds since \code{head} is not reachable from \code{tail}.

Now, assume all invariants hold until a certain point during the execution. We are going to prove that they also hold after executing the next step by one of the system threads.
\begin{enumerate}
    \item If $n_2$ was also $n_1$'s successor before the current step, then by assumption, it is not an infant. Otherwise, $n_1$'s \emph{next} pointer was updated to point to $n_2$ in the current step, either in line~\ref{code.soft_trim_cas} of Listing~\ref{algo:soft_find}, in line~\ref{code.soft_init_state} of Listing~\ref{algo:soft_insert},
    or in line~\ref{code.soft_link_node} of Listing~\ref{algo:soft_insert}.  
    In the first two cases, there exists an earlier point during the execution, in which $n_2$ is the successor of a certain node (during the execution of the find method). By assumption, $n_2$ is not an infant in these cases.
    In the third case, after executing the current step, $n_2$ is not an infant by Definition~\ref{definition:soft_infant}.
    
    \item Assume that the next step will execute line~\ref{code.soft_link_node} of Listing~\ref{algo:soft_insert}, having $newNode=n_2$. Assume by contradiction that $n_2$ is not an infant. Since the CAS in line~\ref{code.soft_link_node} can only be executed on nodes created in line~\ref{code.allocate_new_node}, by the creating thread, $n_2$ is an infant and (1) holds. Now, assume by contradiction that $n_2$'s state is not ``intend to insert''. Then it had been changed in line~\ref{code.soft_complete_insert} of Listing~\ref{algo:soft_insert}, during another insert execution, implying that, by Invariant~\ref{node_inv:no_infant} and the choice of the \code{resultNode} variable, $n_2$ is the successor of some node and thus, is not an infant -- a contradiction. Therefore, $n_2$'s state is ``intend to insert'' and (2) holds as well.
    
    \item If $n_2$ was an infant before the current step, then the current step is the execution of line~\ref{code.soft_link_node} in Listing~\ref{algo:soft_insert}, making $n_2$ the successor of some node which is reachable by assumption. $n_2$ is reachable in this case. Otherwise, by assumption and Claim~\ref{claim:soft_state_changes}, it was reachable during the former step. Assume by contradiction that it is no longer reachable after executing the current step. 
    Then $n_2$ is reachable from a node $n_1$ that was reachable after the previous step, and is no longer reachable (may be $n_2$ itself). Assume w.l.o.g that $n_1$ is such a node for which the path of nodes from Definition~\ref{definition:soft_reachability} is the longest. The node $n_1$ can only become unreachable if the current step is the execution of line~\ref{code.soft_trim_cas} in Listing~\ref{algo:soft_find}, and if $n_1$'s state is ``deleted''. This means that $n_1 \neq n_2$. Since $n_1$'s successor stays reachable in this case, we get a contradiction. Therefore, $n_2$ is reachable in this case as well.
    
    \item By assumption, $n_1$ is not reachable from $n_2$ after the previous step. Since all changes of nodes' successors (line~\ref{code.soft_trim_cas} in Listing~\ref{algo:soft_find} and lines~\ref{code.soft_init_state} and~\ref{code.soft_link_node} in Listing~\ref{algo:soft_insert}) preserve keys order (notice the halting condition in line~\ref{code.soft_find_check_key} of Listing~~\ref{algo:soft_find}), the Invariant still holds.
    
    \item If $n_2$ is not reachable from $n_1$ after the previous step then the invariant holds vacuously. Otherwise, assume by contradiction that $n_2$ was reachable from $n_1$ after the previous step, and is no longer reachable from $n_1$ after the current step. Let $n_3$ be the first node reachable from $n_1$ after the previous step, that is not reachable from it after executing the current step ($n_3$ must exist). 
    The node $n_3$ can only become unreachable from $n_1$ if the current step is the execution of line~\ref{code.soft_trim_cas} in Listing~\ref{algo:soft_find}, and if $n_3$'s state is ``deleted''. This means that $n_3 \neq n_2$. Since $n_3$'s successor stays reachable from $n_1$ in this case, we get a contradiction. Therefore, $n_2$ is still reachable from $n_1$.
    
    \item If $n_1$ was an infant after the previous step then the current step (executing a successful CAS in line~\ref{code.soft_link_node} of Listing~\ref{algo:soft_insert}) makes $n_1$ the predecessor of a node whose \code{tail} is reachable from, by assumption. Therefore, \code{tail} is reachable from $n_1$ in this case. Otherwise, assume by contradiction that \code{tail} was reachable from $n_1$ after the previous step (must hold by assumption), but is no longer reachable from it after the current step. Let $n_2$ be the last node reachable from $n_1$, for whom \code{tail} is not reachable from after executing the current step ($n_2$ must exist).
    Then the current step must change $n_2$'s \emph{next} pointer. Since $n_2$ cannot be an infant (by Invariant~\ref{node_inv:no_infant}), this step is a successful CAS, either in line~\ref{code.soft_trim_cas} of Listing~\ref{algo:soft_find} or in line~\ref{code.soft_link_node} of Listing~\ref{algo:soft_insert}.
    In both cases, $n_2$'s successor is set to be a node that \code{tail} is reachable from, by assumption. Since we get a contradiction to Definition~\ref{definition:soft_reachability}, \code{tail} is reachable from $n_1$ in this case as well.
\end{enumerate}
\end{proof}

\begin{claim} [The Volatile List Invariant] \label{claim:soft_sorted} The volatile list is always sorted by the nodes' keys, no key ever appears twice, and the \code{head} and \code{tail} sentinel nodes are always the first and last members of the list, respectively. 
\end{claim}

\begin{proof}
From Invariant~\ref{node_inv:not_reachable} of Claim~\ref{claim:soft_nodes}, the volatile list is always sorted by the nodes' keys and no key ever appears twice. By Claim~\ref{claim:soft_sentinel_state} and Invariant~\ref{node_inv:is_reachable} of Claim~\ref{claim:soft_nodes}, the \code{head} and \code{tail} sentinel nodes are always members of the list, and by Invariant~\ref{node_inv:not_reachable} of Claim~\ref{claim:soft_nodes}, they are the first and last members, respectively.
\end{proof}

\begin{claim} [Being Logically in the Volatile List] \label{claim:soft_logically} A volatile node $n$ is logically in the list if and only if its state is either ``inserted'' or ``inserted with intention to delete''.
\end{claim}

\begin{proof}
By Definition~\ref{definition:soft_logically}, if $n$ is logically in the list then its state is either ``inserted'' or ``inserted with intention to delete''. 
It remains to show that if its state is either ``inserted'' or ``inserted with intention to delete'' then it is reachable and, thus, logically in the list by Definition~\ref{definition:soft_logically}.
When $n$'s state was changed from ``intend to insert'' to ``inserted'' in line~\ref{code.soft_complete_insert} of Listing~\ref{algo:soft_insert}, it must have had a predecessor. From Invariant~\ref{node_inv:no_infant} of Claim~\ref{claim:soft_nodes}, it is not an infant. From Invariant~\ref{node_inv:is_reachable} of Claim~\ref{claim:soft_nodes}, it is reachable.
\end{proof}
\subsection{Linearizability}\label{chap:soft_linearizability}
We define linearization points for the insert, remove and contains operations, as well as for the find auxiliary method. We explicitly specify the linearization points of the linked-list when no crashes occur.

\subsubsection{Find} \label{chap:soft_find_linearization}
We define the linearization point of the find method to be the point guaranteed from Claim~\ref{claim:soft_dl_find} below.

\begin{claim} \label{claim:soft_dl_find} Let $n_1$ and $n_2$ be the two volatile nodes returned as output from the find method. Then during the method execution, there exist a point in which (1) $n_1$ is reachable, (2) $n_2$ is $n_1$'s successor, and (3) $n_2$'s state is not ``deleted''.
\end{claim}

\begin{proof}
When $n_1$'s state is read for the first time during the execution, it is not ``deleted'' (otherwise, it would have been trimmed and not returned). In addition, since it must have had a predecessor at an earlier point (otherwise, it would not have been traversed), from Invariant~\ref{node_inv:no_infant} of Claim~\ref{claim:soft_nodes}, it is not an infant, and from Invariant~\ref{node_inv:is_reachable} of Claim~\ref{claim:soft_nodes}, it is reachable at this point.
If $n_2$ is $n_1$ successor at this point, then the claim holds for this point. Notice that $n_2$'s state cannot be ``deleted'' at this point, since otherwise, it would have been trimmed at a later point and not returned as output.
If $n_2$ is not $n_1$'s successor at this point, then there exists a point between the first read of $n_1$ and the first read of $n_2$ in which $n_2$ becomes $n_1$'s successor. From Claim~\ref{claim:soft_deleted_state}, $n_1$'s state is not ``deleted'' at this point and thus, from Invariant~\ref{node_inv:is_reachable} of Claim~\ref{claim:soft_nodes}, it is reachable at this point. In addition, $n_2$'s state is not ``deleted'' at this point as well, and the claim holds in this case.
\end{proof}

\subsubsection{Insert} \label{chap:soft_insert_linearization}
Let $n$ be the volatile node created during a successful execution of the insert operation (line~\ref{code.allocate_new_node} in Listing~\ref{algo:soft_insert}). 
Since the operation returns {\em true}, it is guaranteed that $n$'s state changes from ``intend to insert'' to ``inserted'' in line~\ref{code.soft_complete_insert}.
We define the linearization point of a successful insert operation at this point. From Claim~\ref{claim:soft_state_changes} and~\ref{claim:soft_logically}, this is indeed the first point during the execution in which $n$ is logically in the list.

Now, let there be an unsuccessful execution of the insert operation, and let $m$ be the volatile node returned as the second output parameter from the find call in line \ref{code.soft_insert_find}. Since the condition checked in line~\ref{code.soft_node_exists} must hold, its key is equal to the key received as input. Claim~\ref{claim:soft_dl_unsucc_insert} below guarantees that during the execution there exists a point in which $m$ is logically in the list. We set this point as the operation's linearization point in this case.

\begin{claim} \label{claim:soft_dl_unsucc_insert} There exists a point between the linearization point of the mentioned find execution and the return of the operation in which $m$'s state is either ``inserted'' or ``inserted with intention to delete''.
\end{claim}

\begin{proof}
If $m$'s state, read in line~\ref{code.soft_fail_insert}, is not ``intend to insert'', then From Claim~\ref{claim:soft_dl_find} it is guaranteed that at the linearization point of the find execution, $m$'s state is not ``deleted''.
If it is either ``inserted'' or ``inserted with intention to delete'', then from Definition~\ref{definition:soft_logically}, we are done.
Otherwise, it is ``intend to insert''.
However, when checking its state in line~\ref{code.soft_fail_insert}, it is not ``intend to insert'' (since the operation is unsuccessful, and the condition checked in line~\ref{code.soft_fail_insert} must hold). 
From Claim~\ref{claim:soft_state_changes}, it is guaranteed that before checking this condition, there exists a point in which $m$'s state became ``inserted'', and the claim holds.

The remaining case is when the state read in line~\ref{code.soft_fail_insert} is ``intend to insert''.
In this case, the executing thread does not return before $m$'s state changes (the condition checked in line~\ref{code.soft_insert_check_state} holds).
From Claim~\ref{claim:soft_state_changes}, it is guaranteed that there exists a point in which $m$'s state is ``inserted'', and the claim holds in the case as well.
\end{proof}

\subsubsection{Remove} \label{chap:soft_remove_linearization}
Let $n$ be the volatile node returned from the find method call in line~\ref{code.soft_remove_find} of Listing~\ref{algo:soft_remove}.

If the operation returned in line~\ref{code.soft_fail_remove_1_return} then its linearization point is defined at the linearization point of the find call from line~\ref{code.soft_remove_find}. 
The find call returned two nodes that, from Claim~\ref{claim:soft_dl_find}, are guaranteed to be reachable and successive at its linearization point. From Claim~\ref{claim:soft_sorted} it is guaranteed that there does not exist a reachable node with the given key, and in particular, there does not exist a node with the given key which is logically in the list at this point.

If the operation returned in line~\ref{code.soft_fail_remove_2_return}, then the linearization point is the read of \code{currState} during the find execution. Since it was returned from the find call, from Invariant~\ref{node_inv:no_infant} of Claim~\ref{claim:soft_nodes}, it is not an infant. In addition, since its state is ``intend to insert'', from Invariant~\ref{node_inv:is_reachable} of Claim~\ref{claim:soft_nodes}, it is reachable. By Definition~\ref{definition:soft_logically}, it is not logically in the list, and by Claim~\ref{claim:soft_sorted}, there does not exist another node with the given key, which is reachable and in particular, logically in the list at this point.

Otherwise, the operation returned in line~\ref{code.soft_return_result}. It is guaranteed from Claim~\ref{claim:soft_dl_find} that at the linearization point of the find call, $n$'s state was not ``deleted''.
Since the loops in lines \ref{code.soft_while1}--\ref{code.soft_mark_node} and \ref{code.soft_while2}--\ref{code.soft_complete_remove} terminated before the return from the operation in line~\ref{code.soft_return_result}, from Claim~\ref{claim:soft_state_changes}, $n$'s state was changed from ``inserted with intention to delete'' to ``deleted'' at some point between the linearization point of the find method and the return from the operation. This is the operation's linearization point in this case. 
From Claim~\ref{claim:soft_logically}, it is guaranteed that the node stopped being logically in the list exactly at this step.

\subsubsection{Contains} \label{chap:soft_contains_linearization}
Let $n$ be the last volatile node assigned into the \code{curr} variable in line~\ref{code.soft_contains_loop_end} of Listing~\ref{algo:soft_contains}.

\begin{claim} \label{claim:soft_dl_contains} Let $m$ be the last node assigned into the \code{curr} variable before $n$. Then there exists a point during the traversal in which both nodes are reachable and $n$ is $m$'s successor.
\end{claim}

\begin{proof}
Assume by contradiction that the claim does not hold. Let $n_1$ and $n_2$ be the first two nodes for which (1) $n_1$ and $n_2$ are assigned into the \code{curr} variable sequentially, and (2) the guaranteed point does not exist for them. Since this point does exist for $n_1$ and the former node assigned into \code{curr}, $n_1$ is reachable at some point during the execution. From Invariant~\ref{node_inv:still_reachable} of Claim~\ref{claim:soft_nodes}, $n_1$ is reachable as long as its state is not ``deleted''. Since $n_2$ is its successor when assigned into the \code{curr} variable, from Claim~\ref{claim:soft_deleted_state} it was its successor at the last step in which $n_1$ was reachable before this assignment (might be the assignment itself). Therefore, there exists such a point for $n_1$ and $n_2$ -- a contradiction, and the claim holds. 
\end{proof}

If $n$'s key is not equal to the key received as input, then the linearization point is set to be the point guaranteed from Claim~\ref{claim:soft_dl_contains}. From Claim~\ref{claim:soft_sorted}, it is guaranteed that there does not exist a reachable node with the given key at this point. 

Otherwise, $n$'s key is equal to the key received as input. 
If its state, when executing line~\ref{code.read_state}, is either ``inserted'' or ``inserted with intention to delete'', then the operation's linearization point is the read of its state in line~\ref{code.read_state}. 
From Claim~\ref{claim:soft_logically}, $n$ is logically in the list at this point.

If its state is ``intend to insert'' when executing line~\ref{code.read_state}, then the linearization point is set to be the one guaranteed from Claim~\ref{claim:soft_dl_contains}, in which $n$ is reachable. From Claim~\ref{claim:soft_state_changes}, $n$'s state at this point is ``intend to insert'' as well and, thus, it is not logically in the list.
From Claim~\ref{claim:soft_sorted}, since $n$ is reachable, there does not exist another reachable (and in particular, which is logically in the list) node with the given key at this point.

If $n$'s state is ``deleted'' when executing line~\ref{code.read_state} and its state at the point guaranteed from Claim~\ref{claim:soft_dl_contains} is also ``deleted'', then this point is the operation's linearization point. From the above reasons, there does not exist a node with the given key which is logically in the list at this point.

The remaining case is when $n$'s state is not ``deleted'' at the point guaranteed from Claim~\ref{claim:soft_dl_contains}, but it is ``deleted'' when executing line~\ref{code.read_state}. Since its state is eventually ``deleted'', there exists a point between the guaranteed point and the execution of line~\ref{code.read_state} in which $n$ state was changed to ``deleted'' and this is the operation's linearization point in this case.
From Invariant~\ref{node_inv:still_reachable} of Claim~\ref{claim:soft_nodes}, $n$ is reachable at this point and therefore, from the above reasons, there does not exist a node with the given key which is logically in the list at this point in this case as well.
\subsection{Durable Linearizability}\label{chap:soft_dl}
As in Section~\ref{chap:link_free_dl}, we use the notion of \emph{durable linearizability}~\citep{10.1007/978-3-662-53426-7_23} for correctness.
The recovery procedure, executed after a crash (and described in Section~\ref{sub_sec:soft_recovery}), is assumed to terminate before new threads start executing their code.
Given an operation for which a crash event occurs after its invocation and before its response, we consider its response point as the end of the respective recovery procedure.
Notice that in the following definitions, we do not consider recoveries that are interrupted by crash events. We do so for clarity and brevity. The definitions can be easily extended to include such cases.

Before diving into the durable linearizability proof, we prove some basic claims regarding the persistent nodes, used during recovery.

\begin{claim}[State Transitions of Persistent Nodes] \label{claim:soft_persistent_state_changes}
	The state of a persistent node can only go through the following transitions:
	\begin{enumerate}
	    \item From {\em valid} and {\em removed} to {\em invalid}  \label{pstate:from_removed_to_invalid}
	    \item From {\em invalid} to {\em valid} and not {\em removed}  \label{pstate:from_invalid_to_valid}
	    \item From {\em valid} and not {\em removed} to {\em valid} and {\em removed}  \label{pstate:from_valid_to_removed}
	\end{enumerate}
\end{claim}

\begin{proof}
Let $p$ be a persistent node, allocated in line~\ref{code.allocate_new_node} of Listing~\ref{algo:soft_insert}, and let $v$ be the negation of its \code{validStart} bit, when allocated (i.e., $v$ is assigned into the \code{pValidity} field of the respective volatile node).
When $p$ is allocated, its state is {\em valid} and {\em removed}.
The state of $p$ can only change when creating or destroying it (Listing~\ref{algo:soft_pnode_funcs}).
The \code{create} method can only be called from line~\ref{code.soft_init_pnode} of Listing~\ref{algo:soft_insert}, and the \code{destroy} method can only be called from line~\ref{code.flush_delete} of Listing~\ref{algo:soft_remove}, both with $v$ as their \code{pValidity} input parameter.
Notice that the first \code{create} execution terminates before the first \code{destroy} invocation, since the state of the respective volatile node is set to ``inserted'' in line~\ref{code.soft_complete_insert} of Listing~\ref{algo:soft_insert}, only after the termination of the first \code{create} call, and is set to ``inserted with intention to delete'' in line~\ref{code.soft_mark_node} of Listing~\ref{algo:soft_remove}, before the first invocation of the \code{destroy} method (and by Claim~\ref{claim:soft_state_changes}, a volatile node's state can be ``inserted'' only before it becomes ``inserted with intention to delete'').

The first \code{create} execution changes $p$'s state to {\em invalid} and then {\em valid} and not {\em removed}. Any further \code{create} calls do not change its state at all (since the value of the \code{validStart} and \code{validEnd} bits is already $v$). Therefore, any \code{destroy} call can only change it from {\em valid} and not {\em removed} to {\em valid} and {\em removed} (since it only changes the \code{deleted} bit), and the claim follows.
\end{proof}

\begin{claim}[Non-Removed Persistent Nodes] \label{claim:soft_persistent_inserted_nodes}
Let $n$ be a volatile node, and assume its representing persistent node has already been created in line~\ref{code.soft_init_pnode} of Listing~\ref{algo:soft_insert}. If $n$'s state is either ``intention to insert'' or ``inserted'', then the state of its representing persistent node is {\em valid} and not {\em removed}. 
\end{claim}

\begin{proof}
As shown in the proof of Claim~\ref{claim:soft_persistent_state_changes}, the state of $n$'s representing persistent node becomes {\em valid} and not {\em removed} when it is created in line~\ref{code.soft_init_pnode} of Listing~\ref{algo:soft_insert}. In addition, from Claim~\ref{claim:soft_persistent_state_changes}, it can only become {\em valid} and {\em removed}, after $n$'s state becomes ``inserted with intention to delete''. Since $n$'s state is either ``intention to insert'' or ``inserted'', the state of its representing persistent node remains {\em valid} and not {\em removed}.
\end{proof}

\begin{claim}[Removed Persistent Nodes] \label{claim:soft_persistent_removed_nodes}
Let $n$ be a volatile node, and assume its representing persistent node has already been marked as {\em removed} in line~\ref{code.flush_delete} of Listing~\ref{algo:soft_remove}.
Then the state of its representing persistent node does not become  {\em valid} and not {\em removed} anymore.
\end{claim}

\begin{proof}
As shown in Claim~\ref{claim:soft_persistent_state_changes}, any further \code{create} or \code{destroy} calls would not effect the persistent node's state.
\end{proof}

We are going to prove that, given an execution, removing all crash events would leave us with a linearizable history, including all the operations that were fully executed between two crashes, and some of the operations that were halted due to crashes (and then recovered during recovery).
We are going to define, per operation execution, whether it is {\em a surviving operation}. A surviving operation is an operation that is linearized in the final crash-free history of the execution (by removing all crash events). Obviously, operations that were fully executed between two crash events are always considered as surviving operations.
Additionally, we are going to define the linearization points of all surviving operations in the crash-free history.

\subsubsection{Insert}
Before defining the conditions for the survival of an insert operation, we need to re-define the success of an insertion in the presence of crash events.

\begin{definition} [A Successful Insert Operation] \label{definition:soft_dl_successful_insert}
Given an execution of an insert operation, we say that this operation is successful if one of the following holds:
\begin{enumerate}
    \item The operation returns \code{true}.
    \item A volatile node $n$ is allocated in line~\ref{code.allocate_new_node}, the \code{result} variable is assigned with \code{true} in line~\ref{code.soft_insert_result_true}, and the respective persistent node of $n$ is created in line~\ref{code.soft_init_pnode} by some thread before any crash event.
\end{enumerate}
The operation is unsuccessful if it returns \code{false}.
\end{definition}

\begin{definition} [A Surviving Insert Operation] \label{definition:soft_dl_surviving_insert}
An insert operation is considered as a surviving operation if, before the first crash event that follows its invocation, one of the following holds:

\begin{enumerate}
    \item The operation is unsuccessful according to Definition~\ref{definition:soft_dl_successful_insert}. In this case, its linearization point is set to be its original linearization point, presented in Section~\ref{chap:soft_insert_linearization}.
    \item The operation is successful according to Definition~\ref{definition:soft_dl_successful_insert}, and some thread (not necessarily the one that executes the successful insertion) changes the state of the node allocated in line~\ref{code.allocate_new_node}, in line~\ref{code.soft_complete_insert}.
    In this case, the linearization point is set to be its original linearization point as well.
    \item The operation is successful according to Definition~\ref{definition:soft_dl_successful_insert}, and no thread changes the state of the node allocated in line~\ref{code.allocate_new_node}, in line~\ref{code.soft_complete_insert}.
    In this case, the linearization point is set to be the insertion of a new respective volatile node to the list during recovery.
\end{enumerate}
\end{definition}

\begin{claim} \label{claim:soft_dl_insert_linearization}
A surviving insert operation takes effect instantaneously at its linearization point. 
\end{claim}

\begin{proof}
First, let $n$ be the last volatile node allocated during a successful insert operation (according to Definition~\ref{definition:soft_dl_successful_insert}). We are going to show that $n$ is logically inserted into the volatile list at the operation's linearization point (presented in Definition~\ref{definition:soft_dl_surviving_insert}).

If some thread (not necessarily the one that executes the successful insertion) changes the state of $n$ in line~\ref{code.soft_complete_insert}, then by Definition~\ref{definition:soft_dl_surviving_insert}, the operation's linearization point is this change.
As proved in Section~\ref{chap:soft_insert_linearization}, $n$ is indeed logically inserted into the list at this point.
Notice that in this case, we do not consider the insertion of a new representing node during recovery, as a logical insertion of $n$ into the list.
By Invariant~\ref{node_inv:still_reachable} of Claim~\ref{claim:soft_nodes}, $n$ is still reachable when the crash occurs.
In addition, notice that as long as this node is not removed from the list, its state remains ``inserted'' and by Claim~\ref{claim:soft_persistent_inserted_nodes}, the state of its representing persistent node is indeed {\em valid} and not {\em removed} during recovery.

Otherwise, no thread changes $n$'s state from ``intention to insert'' to ``inserted'' before the crash event. 
By Definition~\ref{definition:soft_dl_successful_insert}, a respective persistent node of $n$ is created in line~\ref{code.soft_init_pnode}.
From Claim~\ref{claim:soft_state_changes}, $n$'s state does not change at all before the first crash and therefore, by Definition~\ref{definition:soft_logically}, it is not logically in the list.
As described in Section~\ref{sub_sec:soft_recovery}, during recovery, a new volatile node, representing $n$, is logically inserted into the list, and by Definition~\ref{definition:soft_dl_surviving_insert}, this is the linearization point of the operation in this case.
Notice that by Claim~\ref{claim:soft_persistent_inserted_nodes}, the state of its representing persistent node is indeed {\em valid} and not {\em removed} during recovery, in this case as well.

When an insert operation is unsuccessful by Definition~\ref{definition:soft_dl_successful_insert}, it is also unsuccessful by the original definition.
From Section~\ref{chap:soft_insert_linearization}, there exists a point during its execution for which a node with the given key is already logically in the list and thus, the unsuccessful operation indeed returns a correct answer.
Additionally, notice that even if a representing persistent node is allocated, its state remains {\em valid} and {\em removed}, since the \code{create} method is only called after the volatile node is successfully inserted into the list, and the \code{destroy} method is only called when the state of the volatile node is either ``inserted with intention to delete'' or ``deleted'' (by Claim~\ref{claim:soft_persistent_inserted_nodes}).
\end{proof}

\begin{claim} \label{claim:soft_dl_nonsurviving_insert}
A non-surviving insert operation takes no effect. 
\end{claim}

\begin{proof} 
By Definition~\ref{definition:soft_dl_surviving_insert}, during a none-surviving insert operation, if a volatile node is allocated, and even if it is inserted into the list, its state remains ``intention to insert'' and, thus, it is not logically in the list by Definition~\ref{definition:soft_logically}.
In addition, by Definition~\ref{definition:soft_dl_surviving_insert}, during a non-surviving insert operation, a persistent node may be allocated, but not created (or partially created, and thus, in an {\em invalid} state). 
Therefore, during recovery, even if the persistent node is allocated, its state is either {\em valid} and {\em removed}, or {\em invalid}, and therefore, the represented volatile node is not inserted into the new list.
\end{proof}

\subsubsection{Remove} 
We also re-define the success of a removal in the presence of crash events.

\begin{definition} [A Successful Remove Operation] \label{definition:soft_dl_successful_remove}
Given an execution of an remove operation, we say that this operation is successful if one of the following holds:
\begin{enumerate}
    \item The operation returns \code{true}.
    \item The \code{result} variable is assigned with \code{true} in line~\ref{code.soft_mark_node}, and the respective persistent node is marked as deleted in line~\ref{code.flush_delete} by some thread before any crash event.
\end{enumerate}
The operation is unsuccessful if it returns \code{false}.
\end{definition}

\begin{definition} [A Surviving Remove Operation] \label{definition:soft_dl_surviving_remove}
A remove operation is considered as a surviving operation if, before the first crash event that follows its invocation, one of the following holds:

\begin{enumerate}
    \item The operation is unsuccessful according to Definition~\ref{definition:soft_dl_successful_remove}.
    In this case, its linearization point is set to be its original linearization point, presented in Section~\ref{chap:soft_remove_linearization}.
    \item The operation is successful according to Definition~\ref{definition:soft_dl_successful_remove}, and some thread (not necessarily the one that executes the successful removal) changes the state of the victim node in line~\ref{code.soft_complete_remove}.
    In this case, the linearization point is set to be its original linearization point as well.
    \item The operation is successful according to Definition~\ref{definition:soft_dl_successful_remove}, and no thread changes the state of the victim node in line~\ref{code.soft_complete_remove}.
    In this case, the linearization point is set to be immediately after the crash event (if there is more than one such removal, they are linearized in an arbitrary order).
\end{enumerate}
\end{definition}

\begin{claim} \label{claim:soft_dl_remove_linearization}
A surviving remove operation takes effect instantaneously at its linearization point. 
\end{claim}

\begin{proof}
First, assume a successful remove operation (according to Definition~\ref{definition:soft_dl_successful_remove}), and let $n$ be the node whose state is updated from ``inserted'' to ``inserted with intention to delete'' in line~\ref{code.soft_mark_node}.
We are going to show that $n$ is logically removed from the volatile list at the operation's linearization point (presented in Definition~\ref{definition:soft_dl_surviving_remove}).

If some thread (not necessarily the one that executes the successful removal) changes the state of $n$ from ``inserted with intention to delete'' to ``deleted'' in line~\ref{code.soft_complete_remove}, then by Definition~\ref{definition:soft_dl_surviving_remove}, the operation's linearization point is this change.
As proved in Section~\ref{chap:soft_remove_linearization}, $n$ is indeed logically removed from the list at this point.
Notice that in this case, it is guaranteed that $n$ will not be re-added into the volatile list during recovery, since by Claim~\ref{claim:soft_persistent_removed_nodes}, it has a persistent representative, marked as deleted.

Otherwise, no thread changes $n$'s state from ``inserted with intention to delete'' to ``deleted'' before the crash event. 
By Definition~\ref{definition:soft_dl_successful_remove}, the respective persistent node of $n$ is marked as removed in line~\ref{code.flush_delete}.
From Claim~\ref{claim:soft_state_changes}, $n$'s state remains ``inserted with intention to delete'' until the first crash event.
By Claim~\ref{claim:soft_logically}, it is logically in the list until this crash event.
As described in Section~\ref{sub_sec:soft_recovery}, during recovery, a node representing $n$ will not be inserted into the list (since its representative is marked as deleted, by Claim~\ref{claim:soft_persistent_removed_nodes}), and in particular, will not be reachable.
By Definition~\ref{definition:soft_logically}, it will no longer be logically in the list.
Therefore, it is indeed logically removed from the list at the crash event, right before its linearization point, as presented in Definition~\ref{definition:soft_dl_surviving_remove}.

When a remove operation is unsuccessful by Definition~\ref{definition:soft_dl_successful_remove}, it is also unsuccessful by the original definition.
From Section~\ref{chap:soft_remove_linearization}, there exists a point during its execution for which there is no node with the given key which is logically in the list (and from Claim~\ref{claim:soft_persistent_removed_nodes}, any persistent representative would have a {\em valid} and {\em removed} state) and thus, the unsuccessful operation indeed returns a correct answer.
\end{proof}

\begin{claim} \label{claim:soft_dl_nonsurviving_remove}
A non-surviving remove operation takes no effect. 
\end{claim}

\begin{proof}
By Definition~\ref{definition:soft_dl_surviving_remove}, during a none-surviving remove operation, even if the state of the victim node becomes ``inserted with intention to delete'', by Definition~\ref{definition:soft_dl_surviving_remove}, it does not become ``deleted'', and no thread executes the destruction of its respective persistent node (i.e., by Claim~\ref{claim:soft_persistent_inserted_nodes}, it is still {\em valid} and not {\em removed}).

Therefore, by Definition~\ref{definition:soft_logically}, it is still logically in the list until the crash event occurs, and during recovery, it is re-added to the new list.
\end{proof}

\subsubsection{Contains} 
As opposed to the insert and remove operations, a contains operation is considered as a surviving operation only when it terminates:

\begin{definition} [A Surviving Contains Operation] \label{definition:soft_dl_surviving_contains}
A contains operation is considered as a surviving operation if and only if it terminates before the first crash event occurring after its invocation. 
If it survives, its linearization point is set to be its original linearization point, presented in Section~\ref{chap:soft_contains_linearization}.
\end{definition}

\begin{claim} \label{claim:soft_dl_contains_linearization}
A surviving contains operation takes effect instantaneously at its linearization point. 
\end{claim}

\begin{proof}
Since we only consider contains operations that terminate without being interrupted by crash events, the claim follows directly from Section~\ref{chap:soft_contains_linearization}
\end{proof}

\begin{claim} \label{claim:soft_dl_nonsurviving_contains}
A non-surviving contains operation takes no effect. 
\end{claim}

\begin{proof}
The claim follows directly from the fact that a contains operation (and in particular, an operation with no response), does not change the list.
\end{proof}

\begin{theorem}\label{theorem:soft-durable}
The \SOFT{} list is durable linearizable. 
\end{theorem}

\begin{proof}
By Definition~\ref{definition:soft_dl_surviving_insert}, \ref{definition:soft_dl_surviving_remove} and~\ref{definition:soft_dl_surviving_contains}, all the operations that are fully executed between two crashes (and some of the operations that are halted due to crash events), have a linearization point. 
By Claim~\ref{claim:soft_dl_insert_linearization}, \ref{claim:soft_dl_remove_linearization} and~\ref{claim:soft_dl_contains_linearization}, each operation takes effect instantaneously at its linearization point. By Claim~\ref{claim:soft_dl_nonsurviving_insert}, \ref{claim:soft_dl_nonsurviving_remove} and~\ref{claim:soft_dl_nonsurviving_contains}, operations for which we did not define linearization points (non-surviving operations), do not take effect at all.
In summary, the \SOFT{} list is durable linearizable by definition~\citep{10.1007/978-3-662-53426-7_23}.
\end{proof}

\subsection{Lock-Freedom}\label{chap:soft_lock_freedom}
Similarly to (and following the discussion in) Section~\ref{chap:link_free_lock_freedom}, in this section we prove that in crash-free executions, at least one of the operations terminates.
To derive a contradiction, assume there is some execution for which no executing operation terminates after a certain point.
Notice that we can assume that no operation is invoked after this point, and that the set of running operations is finite (since there is a finite number of system threads).
The rest of the proof relates to the suffix $\alpha$ of the execution, starting from this point.

\begin{claim} \label{claim:no_state_changes}
There is a finite number of state changes during $\alpha$.
\end{claim}

\begin{proof}
An insert operation must terminate after executing line~\ref{code.soft_complete_insert} in Listing~\ref{algo:soft_insert}.
Likewise, a remove operation must terminate after executing line~\ref{code.soft_complete_remove} in Listing~\ref{algo:soft_remove}).
In addition, any remove operation includes at most two successful state changes (in lines~\ref{code.soft_mark_node} and~\ref{code.soft_complete_remove} of Listing~\ref{algo:soft_remove}).
Since the number of running operations is finite by assumption, the number of state changes is finite as well.
\end{proof}

\begin{claim} \label{claim:no_new_reachables}
There is a finite number of pointer changes during $\alpha$.
\end{claim}

\begin{proof}
The loop in lines~\ref{code.soft_insert_start_link_loop}--\ref{code.soft_insert_end_link_loop} of Listing~\ref{algo:soft_insert} must eventually terminate after a successful CAS execution in line~\ref{code.soft_link_node} and therefore, there are no pointer updates in lines~\ref{code.soft_init_state} and~\ref{code.soft_link_node} of Listing~\ref{algo:soft_insert}. 
Thus, pointer updates can only occur in line~\ref{code.soft_trim_cas} of Listing~\ref{algo:soft_find}.
When executing this update, the \code{pred} node is reachable from Claim~\ref{claim:soft_deleted_state} and Invariants~\ref{node_inv:no_infant} and~\ref{node_inv:is_reachable} of Claim~\ref{claim:soft_nodes}.
Since \code{curr} is \code{pred}'s successor and \code{succ} is \code{curr}'s successor right before this change, \code{succ} is also reachable before this step.
Therefore, no node becomes reachable when executing this CAS.
Since this is the only possible pointer change, the list can only shrink, and the number of such pointer changes is finite.
\end{proof}

\begin{theorem} \label{theorem:soft_lock_free}
The \SOFT{} list is lock-free. 
\end{theorem}

\begin{proof}
From Claims~\ref{claim:no_state_changes} and~\ref{claim:no_new_reachables}, after a certain point, there are no state or pointer changes.
Therefore, we consider the suffix $\alpha'$ of the execution that contains no state or pointer changes.
Obviously, starting from this point, the list becomes stable, and does not change anymore.

Since the list is finite, from Claim~\ref{claim:soft_sorted}, every find and contains execution eventually ends.
In addition, every insert and remove operation must be unsuccessful, and also terminate (since calls to the find method always terminate).
We get a contradiction and therefore, the implementation is lock-free.
\end{proof}

\end{document}